\documentclass{article}
\usepackage{arxiv}
\usepackage{microtype}
\usepackage{graphicx}
\usepackage{subfigure}
\usepackage{booktabs} % for professional tables
\usepackage{hyperref}

\usepackage{amsmath}
\usepackage{amssymb}
\usepackage{mathtools}
\usepackage{amsthm}
\usepackage[]{empheq}
\usepackage{thm-restate}
\usepackage{algorithm}
\usepackage{algorithmic}

\usepackage{xcolor}         % colors
\usepackage{nicefrac, xfrac}

\usepackage{natbib}
\newcommand{\D}{\mathcal{D}}
\newcommand{\NPHard}{$\mathsf{NP}$-hard}

\newcommand{\A}{\mathcal{A}}

\newcommand{\N}{\mathcal{N}}

\newcommand{\ie}{\emph{i.e., }}

\newcommand{\mcS}{\mathcal{S}}

\newcommand{\evec}{\boldsymbol{e}}

\newcommand{\avec}{\boldsymbol{a}}

\newcommand{\Reals}{\mathbb{R}}
\newcommand{\mfR}{\mathfrak{R}}
\newcommand{\etavec}{\boldsymbol{\eta}}

\newcommand{\xvec}{\boldsymbol{x}}

\newcommand{\yvec}{\boldsymbol{y}}
\newcommand{\zvec}{\boldsymbol{z}}

\newcommand{\zerovec}{\boldsymbol{0}}
\newcommand{\onevec}{\boldsymbol{1}}

\newcommand{\cvec}{\boldsymbol{c}}

\newcommand{\argmin}{\operatornamewithlimits{argmin}}

\newtheorem{example}{Example}

\usepackage{pifont}
\usepackage[normalem]{ulem} % for \sout

\definecolor{mygreen}{rgb}{0.0, 0.5, 0.0}
\definecolor{myorange}{rgb}{0.55, 0.62, 1}

\allowdisplaybreaks

\usepackage[capitalize,noabbrev]{cleveref}

\usepackage{xparse}

\theoremstyle{plain}
\newtheorem{theorem}{Theorem}[section]
\newtheorem{proposition}{Proposition}[section]
\newtheorem{lemma}{Lemma}[section]
\newtheorem{corollary}[theorem]{Corollary}
\theoremstyle{definition}
\newtheorem{definition}{Definition}[section]
\newtheorem{assumption}{Assumption}
\theoremstyle{remark}

\usepackage[textsize=tiny]{todonotes}

\title{Constrained Phi-Equilibria}

\author{
	Martino Bernasconi\\
	Politecnico di Milano\\
	\texttt{martino.bernasconideluca@polimi.it}
	\And
	Matteo Castiglioni\\
	Politecnico di Milano\\
	\texttt{matteo.castiglioni@polimi.it}
	 \And
	Alberto Marchesi\\
	Politecnico di Milano\\
	\texttt{alberto.marchesi@polimi.it}
	\And
	Francesco Trovò\\
	Politecnico di Milano\\
	\texttt{francesco1.trovo@polimi.it}
	\And
	Nicola Gatti\\
	Politecnico di Milano\\
	\texttt{nicola.gatti@polimi.it}
}

\begin{document}
	
\maketitle

% Abstract. Note that this must come before \maketitle.

	% this must go after the closing bracket ] following \twocolumn[ ...
	
	% This command actually creates the footnote in the first column
	% listing the affiliations and the copyright notice.
	% The command takes one argument, which is text to display at the start of the footnote.
	% The \icmlEqualContribution command is standard text for equal contribution.
	% Remove it (just {}) if you do not need this facility.
	
	%\printAffiliationsAndNotice{}  % leave blank if no need to mention equal contribution

	\begin{abstract}
		The computational study of equilibria involving constraints on players' strategies has been largely neglected.
		However, in real-world applications, players are usually subject to constraints ruling out the feasibility of some of their strategies, such as, \emph{e.g.}, safety requirements and budget caps.
		Computational studies on constrained versions of the Nash equilibrium have lead to some results under very stringent assumptions, while finding constrained versions of the \emph{correlated equilibrium} (CE) is still unexplored.
		%, even though it is much more computationally.
		%
		In this paper, we introduce and computationally characterize \emph{constrained Phi-equilibria}---a more general notion than constrained CEs---in \emph{normal-form games}.
		We show that computing such equilibria is in general computationally intractable,
		%raises considerable computational challenges, 
		and also that the set of the equilibria may \emph{not} be convex, providing a sharp divide with unconstrained CEs.
		Nevertheless, we provide a polynomial-time algorithm for computing a constrained (approximate) Phi-equilibrium maximizing a given linear function,
		%can be done in polynomial time, 
		when either the number of constraints or that of players' actions is fixed.
		Moreover, in the special case in which a player's constraints do \emph{not} depend on other players' strategies, we show that an exact, function-maximizing equilibrium can be computed in polynomial time, while \emph{one} (approximate) equilibrium can be found with an efficient decentralized no-regret learning algorithm.
\end{abstract}
	\section{Introduction}

Over the last years, equilibrium computation problems have received a terrific attention from AI and ML research~\citep{brown2019superhuman,meta2022human}, as game-theoretical equilibrium notions provide a principled framework to deal with multi-player decision-making problems.
Most of the works on equilibrium computation problems focus on classical solution concepts---such as the well-known \emph{Nash equilibrium} (NE)~\citep{nash1951non} and \emph{correlated equilibrium} (CE)~\citep{aumann1974subjectivity}---, thus neglecting the presence of constraints entirely.
However, in most of the real-world applications, the players are usually subject to constraints that rule out the feasibility of some of their strategies, such as, \emph{e.g.}, safety requirements and budget caps.
%
%
%However, the most widely studied notions of equilibrium---such as the well-known \emph{Nash equilibrium} (NE)~\citep{nash1951non} and \emph{correlated equilibrium} (CE)~\citep{aumann1974subjectivity}---entirely neglect the presence of constraints. 
%
Thus, addressing equilibrium notions involving constraints is a crucial step needed for the operationalization of game-theoretic concepts into real-world settings.

The study of equilibrium notions involving constraints was initiated by~\citet{arrow1954existence}, who define the concept of \emph{generalized} NE (GNE).
The GNE can be interpreted as an NE of a game where players' strategies are subject to some constraints, which must be satisfied at the equilibrium and also determine which are the feasible players' deviations.
%
% The computational study of GNEs has been initiated by~\citet{rosen1965existence}, and~\citep{kanzow2016augmented,bueno2019optimality,jordan2022first} are some of the most recent works on the topic (see~\citep{facchinei2010generalized} for a survey of earlier works).
%
% and, since then, many works addressed the problem of computing GNEs by mainly exploiting techniques based on \emph{quasi-variational inequalities} (see~\citep{facchinei2010generalized} for a survey).
%
However, given that computing a GNE is clearly $\mathsf{PPAD}$-hard~\citep{daskalakis2009complexity}, all the works dealing with the computation of GNEs (see, \emph{e.g.},~\citep{facchinei2010generalized}) provide efficient algorithms only in specific settings that require very stringent assumptions. 
%
% More recently, some works~\citep{kanzow2016augmented,bueno2019optimality,jordan2022first} also studied the convergence of learning dynamics  

Most of the computationally challenges in finding GNEs are inherited from the NE.
In settings in which constrained are \emph{not} involved, the computational issues of NEs are usually circumvented by considering weaker equilibrium notions.
Among them, those that have received most of the attention in the literature are the CE and its variations, which have been shown to be efficiently computable in several settings of interest~\citep{papadimitriou2008computing,celli2020no}.
Surprisingly, with the only exception of~\citep{chenfinding} (see Section~\ref{sec:related} for a detailed discussion on it), no work has considered the problem of computing CEs in constrained settings.
Thus, investigating whether the CE retains its nice computational properties when adding constraints on players' strategies is an open interesting question.

\subsection{Original Contributions}

In this paper, we introduce and computationally characterize \emph{constrained Phi-equilibria}, starting, as it is customary, from the setting of \emph{normal-form games}.
Our equilibria include the constrained versions of the classical CE and all of its variations as special cases, by generalizing the notion of Phi-equilibria introduced by~\citet{greenwald2003general} to constrained settings.
In particular, constrained Phi-equilibria are defined as Phi-equilibria, but in games where players are subject to some constraints.
Such constraints must be satisfied at the equilibrium, and, additionally, players are only allowed to undertake \emph{safe deviations}, namely those that are feasible according to the constraints.
%
%by players' deviations in order for them to be safe (\emph{i.e.}, feasible according to the constraints).
%
Crucially, the set of safe deviations of a player does \emph{not} only depend on the strategy of that player, but also on those of the others.

We start by showing that one of the most appealing computational properties of Phi-equilibria, namely that the set of the equilibria of a game is convex, is lost when moving to their constrained version.
This raises considerable computational challenges in computing constrained Phi-equilibria.
Indeed, we formally prove a strong intractability result: for any factor $\alpha > 0$, it is \emph{not} possible, unless $\mathsf{P} = \mathsf{NP}$, to find in polynomial time a constrained (approximate) Phi-equilibrium which achieves a multiplicative approximation $\alpha$ of the optimal value of a given linear function.
Then, in the rest of the paper, we show several ways in which such a negative result can be circumvented.

We prove that a constrained approximate Phi-equilibrium which maximizes a given linear function can be found in polynomial time, when either the number of constraints or that of players' actions is fixed.
Our results are based on a general algorithm that employs a non-standard ``Lagrangification'' of the constraints defining the set of safe deviations of a player.
%, and a clever discretization 
%
Moreover, the algorithm needs a way of dealing with the non-convexity of the set of the equilibria, which we provide in the form of a clever discretization of the space of the Lagrange multipliers.

Finally, we focus on the special case in which the constraints defining the safe deviations of a player do \emph{not} depend on the the strategies of the other players, but only on the strategy of that player.
This includes constrained Phi-equilibria identifying a particular constrained version of the \emph{coarse} CE by~\citet{moulin1978}, in which the players' strategies are subject to \emph{marginal cost constraints}.
These arise in several real-world applications in which the players have bounded resources, such as, \emph{e.g.}, budget-constrained bidding in auctions.
In such a special case, we prove that a constrained (exact) Phi-equilibrium maximizing a given linear function can be computed in polynomial time, and we provide an efficient decentralized no-regret learning algorithm for finding \emph{one} constrained (approximate) Phi-equilibrium.

%We conclude the paper by introducing a challenging open question related to the problem studied in this paper, paving the way for future research on the topic.

\subsection{Related Works}\label{sec:related}

%Next, we survey the works that are the most related to ours.

\paragraph{GNEs}
\citet{rosen1965existence} initiated the study of the computational properties of GNEs.
%
% introduced by~\cite{arrow1954existence} to study the equilibria of abstract economies. GNE can be seen as a generalization of Nash equilibria in constrained games in which both player's strategies and best responses are constrained to satisfy some costs constraints.
%
After that, several other works addressed the problem of computing GNEs by mainly exploiting techniques based on \emph{quasi-variational inequalities} (see~\citep{facchinei2010generalized} for a survey).
%
%several other works deepened the understanding of GNE in games, see~\citet{facchinei2010generalized} for a survey on the earlier works, mainly based on \emph{quasi-variational inequalities}.
%
More recently, some works~\citep{kanzow2016augmented,bueno2019optimality,jordan2022first,goktasexploitability} also studied the convergence of iterative optimization algorithms to GNEs.
%
%More recently the following works~\cite{kanzow2016augmented, bueno2019optimality, jordan2022first} addressed the convergence guarantees of learning GNE with the common assumption of the constrain space being jointly convex.
%
In order to provide efficient algorithms, all these works need to introduce very stringent assumptions, which are even stronger than those required for the efficient computation of NEs.

\paragraph{Constrained Markov Games}
Equilibrium notions involving constraints have also been addressed in the literature on \emph{Markov games}, with~\citep{altman2000constrained, alvarez2006existence} being the first works introducing GNEs in such a field.
More recently, \citet{hakami2015learning} defined a notion of constrained CE in Markov games.
However, the incentive constraints in their notion of equilibrium only predicate on ``pure'' deviations, which, in presence of constraints, may lead to empty sets of safe deviations.
Very recently,~\citet{chenfinding} generalize the work of~\citet{hakami2015learning} by considering ``mixed'' deviations.
However, their algorithm provides rather weak convergence guarantees, as it only ensures that the returned solution satisfies incentive constraints in expectation.
Indeed, as we show in Proposition~\ref{prop:non-convex}, the set of constrained equilibria may \emph{not} be convex (it is easy to see that Example~\ref{example} also applies to the setting studied by~\citet{chenfinding}), and, thus, the fact that incentive constraints are only satisfied in expectation does \emph{not} necessarily imply that the ``true'' incentive constraints defining the equilibrium are satisfied.
We refer the reader to Appendix~\ref{app:weakness} for additional details on these aspects.
%
%
%
%Concurrently, constrained Nash equilibria where studied in the Safe Reinforcement Learning community, starting from~\citet{altman2000constrained, alvarez2006existence} these work introduces the concept of GNE in the field of Markov Games.
%More recently, the study of equilibria in Markov Games was extended to the concept of Correlated Equilibria~\citep{aumann1974subjectivity} defining Constrained Markov Games~\citep{hakami2015learning,chenfinding}. Specifically~\citet{chenfinding} is the only work that considers solution concepts other the Nash in Constrained Markov Games. However, their algorithm only finds a relaxed solution concept, which satisfies the correlated equilibria incentives only in expectation.

%
%Anther line of work connected to this paper is the one that studies Phi-equilibria. This line of study was introduced by~\citet{greenwald2003general} and~\citet{stoltz2007learning}. They generalize the concept of Correlated Equilibria and unify many other solution concepts, such as Coerced Correlated Equilibria.

	\section{Preliminaries}\label{sec:prelim}

In this section, we introduce all the preliminary definitions and results that are needed in the rest of the paper.

\subsection{Cost-constrained Normal-form Games}

%We let $\N := \{1,\ldots, n\}$ be the set of players 
%
In a \emph{normal-form game}, there is a finite set $\N := \{1,\ldots, n\}$ of $n$ players.
Each player $i \in \N$ has a finite set $\A_i$ of actions available, with $s := |\A_i|$ for $i \in \N$ being the number of players' actions.\footnote{For ease of presentation, in this paper we assume that all the players have the same number of actions. All the results can be easily generalized to the case of different numbers of actions.}
We denote by $\avec \in \A := \bigtimes_{i \in \N} \A_i$ an action profile specifying an action $a_i$ for each player $i \in \N$.
Moreover, for $i \in \N$, we let $\avec_{-i} \in \A_{-i} := \bigtimes_{j \neq i \in \N} \A_i$ be an action profile of all players other than $i$, while $(a,\avec_{-i})$ is the action profile obtained by adding $a \in \A_i$ to $\avec_{-i}$.
%
%
%Moreover, we denote the set of actions of all players except player $i$ by $A^{-i}:=\times_{j\in[N], j\neq i}\A^j$ and the set of all possible action profiles as $\A:=\times_{j\in[N]}\A^j$.
%%
%As common in normal form games we denote as $\Delta(B)$ the set of distributions over any finite set $B$
%
Finally, we let $u_i:\A\to[0,1]$ be the utility function of player $i\in \N$, with $u_i(\avec)$ being the utility perceived by that player when the action profile $\avec\in\A$ is played.

%We denote by $\zvec \in \Delta_{\A}$ a \emph{correlated strategy} defined over the set of all the actions profiles, with $\zvec[\avec]$ denoting the probability assigned to $\avec \in \A$.\footnote{In this paper, given a finite set $X$, we denote by $\Delta_X$ the set of all the probability distributions defined over the elements of $X$.}
%%
%% We let $u_i:\A\to[0,1]$ be the utility perceived by player $i\in N$ when an action profile $\avec\in\A$ is played.
%%
%Then, with an abuse of notation, for every player $i\in \N$, we use $u_i(\zvec)$ to denote player $i$'s expected utility when the action profile played by the players is drawn according to a distribution $\zvec \in \Delta_{\A}$.
%%
%In particular, it holds $u_i(\zvec):=\sum_{\avec\in\A}u_i(\avec)  \zvec[\avec]$.
%
% $U_i:\Delta(\A)\to[0,1]$ be the expected utility perceived by the $i$-th player when actions for all players are drown from the a correlated distribution in $\Delta(\A)$. Formally, this is defined as $U_i(\zvec):=\sum_{\avec\in\A}u_i(\avec)\zvec[\avec]$.

We extend classical normal-form games by considering the case in which each player $i \in \N$ has $m_i$ \emph{cost functions}, namely $c_{i,j}: \A \to [-1,1]$ for $j \in [m_i]$.\footnote{In this paper, given some $x \in \mathbb{N}_{>0}$, we let $[x] := \{ 1,\ldots,x \}$ be the set of the first $x$ natural numbers.}
Each player $i \in \N$ is subject to $m_i$ \emph{constraints}, which require that all player $i$'s costs are less than or equal to zero.\footnote{Since $\zvec\in\Delta_{\A}$, we can assume w.l.o.g. that all the constraints are of the form $\leq 0$, as any constraint can always be cast in such a form by suitably manipulating the cost function $c_{i,j}$.}
%bounding their expected costs.
%
% W.l.o.g., we assume that such constraints require that .\footnote{Given that $\zvec\in\Delta_{\A}$, a constraint can always be expressed as $\preceq \zerovec$ by adding suitable constants to each cost function $c_{i,j}$.}
%
For ease of notation, we assume w.l.o.g. that all players have the same number of constraints, namely $m := m_i$ for all $i \in \N$.
Moreover, we encode the costs of player $i \in \N$ by a vector-valued function $\cvec_i:\A\to[-1,1]^{m}$ such that, for every $\avec \in \A$, the $j$-th component of the vector $\cvec_i(\avec)$ is $c_{i,j}(\avec)$.

\paragraph{Correlated Strategies}
In this paper, we deal with solution concepts defined by correlated strategies.
A \emph{correlated strategy} $\zvec \in \Delta_{\A}$ is a probability distribution defined over the set of actions profiles, with $\zvec[\avec]$ denoting the probability assigned to $\avec \in \A$.\footnote{In this paper, given a finite set $X$, we denote by $\Delta_X$ the set of all the probability distributions defined over the elements of $X$.}
%
% We let $u_i:\A\to[0,1]$ be the utility perceived by player $i\in N$ when an action profile $\avec\in\A$ is played.
%
With an abuse of notation, for every player $i\in \N$, we let $u_i(\zvec)$ be player $i$'s expected utility when the action profile played by the players is drawn from $\zvec \in \Delta_{\A}$.
In particular, it holds $u_i(\zvec):=\sum_{\avec\in\A}u_i(\avec)  \zvec[\avec]$.
Similarly, we let $\cvec_i(\zvec):=\sum_{\avec\in\A}\cvec_i(\avec)\zvec[\avec]$ be the vector of player $i$'s expected costs,
% under a correlated strategy $\zvec \in \Delta_{\A}$.
%
so that player $i$'s constraints can be compactly written as $\cvec_i(\zvec) \preceq \zerovec$.
Finally, we define $\mcS \subseteq \Delta_{\A}$ as the set of \emph{safe} correlated strategies, which are those satisfying the cost constraints of all players.
Formally:
\[
	\mcS_i:=\left\{ \zvec\in\Delta_{\A} \mid \cvec_i(\zvec)\preceq \zerovec \right\}, \text{ and } \mcS:=\bigcap_{i\in\N} \mcS_i.
\]
In the following, we assume w.l.o.g. that $\mcS \neq \varnothing$.

\subsection{Constrained Phi-equilibria}

We generalize the notion of Phi-equilibria~\citep{greenwald2003general} to cost-constrained normal-form games.
Such equilibria are defined as correlated strategies $\zvec \in \Delta_{\A}$ that are robust against a given set $\Phi$ of players' deviations, in the sense that, if a mediator draws an action profile $\avec \in \A$ according to $\zvec$ and recommends to play action $a_i$ to each player $i \in \N$, then no player has an incentive to deviate from their recommendation by selecting a deviation in $\Phi$.

For every $i \in \N$, we let $\Phi_i$ be the set of player $i$'s \emph{deviations}, \emph{i.e.}, linear transformations $\phi_i : \A_i \to \Delta_{\A_i}$ that prescribe a probability distribution over player $i$'s actions for every possible action recommendation.
For ease of notation, we encode a deviation $\phi_i$ by means of its matrix representation.
Formally, an entry $\phi_i[b,a]$ of the matrix represents the probability assigned to action $a \in \A_i$ by $\phi_i(b)$.
We denote the set of all the possible deviations by $\Phi := \{\Phi_i\}_{i\in \N}$.

Given a correlated strategy $\zvec \in \Delta_{\A}$ and a deviation $\phi_i \in \Phi_i$, we define $\phi_i\diamond \zvec$ as the \emph{modification} of $\zvec$ induced by $\phi_i$, which is a linear transformation that can be expressed as follows in terms of matrix representation:
\[
	(\phi_i\diamond \zvec)[a_i,\avec_{-i}] := \sum\limits_{b\in\A_i}\phi_i[b,a_{i}]\zvec[b,\avec_{-i}],
\]
for every $a_i \in \A_i$ and $\avec_{-i} \in \A_{-i}$.
Moreover, given a set $\Phi_i$ of deviations of player $i \in \N$, in the following we denote by $\Phi_{i}^{\mcS}(\zvec) := \left\{ \phi_i\in\Phi_i \mid \phi_i\diamond \zvec\in \mcS_i \right\}$ the set of \emph{safe deviations} for player $i$ at a given correlated strategy $\zvec\in\Delta_{\A}$.\footnote{In the conference version \citep{bernasconi2023constrained} of this paper we incorrectly defined the set $\Phi_i^\mcS$ as the set of deviations that map $\zvec$ into $\mcS$, \ie $\Phi_{i}^{\mcS}(\zvec) = \left\{ \phi_i\in\Phi_i \mid \phi_i\diamond \zvec\in \mcS \right\}$. However all the results of the conference version, used the current, correct definition.}
%
%Let $\Phi$ be the set of possible deviations, i.e. $\Phi=\{\Phi_i\}_{i\in[N]}$, where each $\Phi_i$ is a subset of linear transformations $\phi_i:\A^i\to\Delta(\A^i)$ that prescribe a deviation from an action $a\in\A^i$ to a distribution over the actions $\A^i$. For ease of notation we represent a deviation $\phi_i$ with its matrix representation. More formally, we write $\phi_i[b,a]$ be the probability of playing action $a$ when action $b$ is recommanded. With this notation the modification $\phi_i\diamond \zvec$ is a linear function and can be expressed as:
%\[
%(\phi_i\diamond \zvec)[a_i,\avec_{-i}] := \sum\limits_{b\in\A^i}\phi_i[b,a_{i}]\zvec[b,\avec_{-i}].
%\]

We are now ready to provide our definition of \emph{constrained Phi-equilibria} in cost-constrained normal-form games.
\begin{definition}[Constrained $\epsilon$-Phi-equilibria]
	% Given a set $\Phi := \{\Phi_i\}_{i\in \N}$ of deviations, let $\Phi_{i}^{\mcS}(\zvec) := \left\{ \phi_i\in\Phi_i \mid \phi_i\diamond \zvec\in \mcS \right\}$ be the of constrained deviations at any correlated strategy $\zvec\in\Delta_{\A}$.
	%
	Given a set $\Phi := \{ \Phi_i \}_{i \in \N}$ of deviations and an $\epsilon>0$, a \emph{constrained $\epsilon$-Phi-equilibrium} is a safe correlated strategy $\zvec\in\mcS$ such that, for all $i\in \N$, the following holds:
	\[
	 u_i(\zvec)\ge u_i(\phi_i\diamond \zvec)-\epsilon\quad\forall\phi_i\in \Phi^\mcS_i(\zvec).
	\]
	A \emph{constrained Phi-equilibrium} is defined for $\epsilon = 0$.
	%	
	%	Let $\Phi_{i}^{\mcS}(\zvec):=\{\phi_i\in\Phi_i\,|\,\phi_i\diamond \zvec\in \mcS\}$ be the set of constrained deviations at a specific correlated strategy $\zvec\in\Delta(\A)$.
	%	%
	%	For any $\epsilon>0$, a Constrained $\epsilon$-Phi-equilibrium is defined as a correlated strategy $\zvec\in\mcS$ such that for all $i\in[N]$:
	%	%
	%	\[
	%	U_i(\zvec)\ge U_i(\phi_i\diamond \zvec)-\epsilon\quad\forall\phi_i\in \Phi^\mcS_i(\zvec).
	%	\]
\end{definition}

\subsection{Computing Constrained Phi-equilibria}

In the following, we formally introduce the computational problem that we study in the rest of the paper.

We denote by $I := (\Gamma,\Phi)$ an instance of the problem, where the tuple $\Gamma := (\N,\A,\{u_i\}_{i\in \N}, \{c_{i,j} \}_{i\in \N,j\in [m]})$ is a cost-constrained normal-form game and $\Phi := \{ \Phi_i \}_{i \in \N}$ is a set of deviations.
Moreover, we let $|I|$ be the size (in terms of number of bits) of the instance $I$.
We assume that the number $n$ of players is fixed, so that $|I|$ does \emph{not} grow exponentially in $n$.\footnote{Notice that the size of the representation of a normal-form game is $O(s^n)$, and, thus, exponential in $n$. Any algorithm that runs in time polynomial in such instance size is \emph{not} computationally appealing, as even its input has size exponential in $n$. For this reason, we focus on the case in which $n$ is fixed, and, thus, the instance size does \emph{not} grow exponentially with $n$.}
We also make the following assumption on how the sets of deviations are represented:
% specified in the representation of an instance $I$.
%
% We first introduce a very natural assumption on the set of possible deviations $\Phi_i$.
%
\begin{assumption}\label{ass:solvablePolytope}
	For every $i \in \N$, the set $\Phi_i$ is a polytope encoded by a finite of linear inequalities.\footnote{Notice that, since each $\phi_i \in \Phi_i$ is represented as a matrix, a linear inequality is expressed as $\sum_{b,a \in \A_i} M[b,a] \phi_i[b,a]\le d$, for some matrix $M$ and scalar $d$.}
	%	
	%	$\Phi_i$ is a polytope for all $i\in[N]$ encoded by a finite set of linear inequalities of the form $\sum_{a,b} M[a,b] \phi[a,b]\le d$ for two given M and d.
\end{assumption}
Let us remark that, in games without constraints, this assumption is met by all the sets $\Phi$ which determine the classical notions of Phi-equilibria~\citep{greenwald2003general}.
%
%We remark that in the setting without constrains, this assumption is met by all the  sets $\Phi$ which give arise to all the main solution concepts~\cite{greenwald2003general}.

Next, we formally define our computational problem:
\begin{definition}[\textsc{ApxCPE}$(\alpha,\epsilon)$]\label{def:probelm}
	For any $\alpha,\epsilon > 0$, we define \textsc{ApxCPE}$(\alpha,\epsilon)$ as the problem of finding, given an instance $I := (\Gamma,\Phi)$ and a linear function $\ell : \Delta_{\A} \to \mathbb{R}$ as input, a constrained $\epsilon$-Phi-equilibrium $\zvec \in \Delta_{\A}$ such that $\ell(\zvec)\ge \alpha\ell(\zvec')$ for all constrained Phi-equilibria $\zvec' \in \Delta_{\A}$.
	%	
	%	Given a set of deviations for each player $\Phi:=\{\Phi_i\}_{i\in[N]}$, a linear function $\ell$, and two values $\alpha, \epsilon\in [0,1]$, \textsc{ApxCPE}$(\alpha,\epsilon)$ is the problem of finding a Costrained $\epsilon$-Phi-equilibrium $\zvec$ such that $\ell(\zvec)\ge \alpha\ell(\zvec^\star)$, where $\zvec^*$ is the optimal Constrained Phi-equilibrium for the linear function $\ell$.
	%	%\ma{definire $OPT$... Per ora ho lasciato $L_{\D,0}$.}
\end{definition}
Intuitively, \textsc{ApxCPE}$(\alpha,\epsilon)$ asks to compute a constrained $\epsilon$-Phi-equilibrium whose value for the linear function $\ell$ is at least a fraction $\alpha$ of the maximum value which can be achieved by an (exact) constrained Phi-equilibrium.

In order to ensure that an instance of our problem is well defined, we make the following ``Slater-like'' assumption on how the players' cost constraints are defined.
\begin{assumption} \label{ass:strictly}
	For every correlated strategy $\zvec\in\Delta_{\A}$, player $i\in \N$, and index $j\in[m]$, there exists $\phi^\circ_i\in\Phi_i^\mcS(\zvec)$:
	\[ 
		c_{i,j}(\phi^\circ_i \diamond \zvec)\le -\rho,
	\]
	where $\rho>0$ and $1/\rho$ is $O(\textnormal{poly}(|I|))$, with $\textnormal{poly}(|I|)$ being a polynomial function of the instance size $|I|$.
\end{assumption}
In Assumption~\ref{ass:strictly}, the condition $\rho > 0$ is required to guarantee the existence of a constrained Phi-equilibrium (see Theorem~\ref{thm:existence}) and that the sets $\Phi_i^\mcS(\zvec)$ are non-empty (otherwise our solution concept would be ill defined).
Moreover, the second condition on $\rho$ in Assumption~\ref{ass:strictly} is equivalent to requiring that our algorithms run in time polynomial in $\frac{1}{\rho}$.
%
%\textcolor{red}{Bisogna mettere dei cetrioli sulla seconda assunzione.}
% \textcolor{red}{Moreover, the second condition on $\rho$ is the weaker assumption needed to avoid some technicalities in the derivation of our results.}
%
%Assumption~\ref{ass:strictly} is required to prove the existence of Constrained Phi-equilibrium and that the set $\Phi_i^\mcS(\zvec)$ is non-empty. Moreover, the polynomial relation between $1/\rho$ and the instance size $|\I|$ is the looser assumption required to avoid computational technicalities and to avoid that $\rho\to0$ for which the solution concept might be ill defined. 

Assumption~\ref{ass:strictly} also allows us to prove the existence of our equilibria, by showing that the constrained Nash equilibria introduced by~\citet{altman2000constrained}, which always exist under Assumption~\ref{ass:strictly}, are also constrained Phi-equilibria.
%
% With these assumptions we can readily prove the existence of our solution concept.
%
\begin{restatable}{theorem}{existence}\label{thm:existence}
	Given a cost-constrained normal-form game $\Gamma$ and a set $\Phi$ of deviations, if Assumption~\ref{ass:strictly} is satisfied, then $\Gamma$ admits a constrained Phi-equilibrium. 
	%
	% There exists a Constrained Phi-equilibrium for any deviation set $\Phi$ and costs $\{\Cvec_i\}_{i\in[N]}$ such that Assumption~\ref{ass:strictly} holds.
\end{restatable}

\subsection{Relation with Unconstrained Phi-equilibria}\label{sec:prelim_relations}

We conclude the section by discussing the relation between our constrained Phi-equilibria and classical equilibrium concepts for unconstrained games.
%
% In particular, our equilibria include the constrained versions of classical equilibrium notions based on correlation.
%
%In this section we discuss the relation and differences of our definition to the standard definition of Phi-equilibria.
%
% First we introduce the two most common sets of deviations that include many common definition of equilibria in Normal Form Games in absence of costs constraints.

\paragraph{Correlated Equilibrium}
When there are no constraints, the \emph{correlated equilibrium} (CE)~\citep{aumann1974subjectivity} is a special case of Phi-equilibrium.
As shown by~\citet{greenwald2003general}, the CE is obtained when the sets $\Phi_i$ contain all the possible deviations. 
Formally, the CE is defined by the set $\Phi_{\textsc{all}} := \{ \Phi_{i,\textsc{all}} \}$ of deviations such that:
\[
	\Phi_{i,\textsc{all}} := \left\{\phi_i  \, \Big\vert \, \sum_{a\in\A_i}\phi_i[b,a]=1 \quad \forall b\in\A_i\right\}.
\]
%
%
%
%\[
%\Phi_{i,\textnormal{ALL}} := \left\{\phi_i[b,a]\in[0,1]: \sum_{a\in\A^i}\phi_i[b,a]=1,\, \forall b\in\A^i\right\}.
%\]
%This set encodes all the possible deviations strategies. In absence of costs constraints, the set of Phi-equilibria when using this set of deviations are equivalent to the set of \emph{Correlated Equilibria}~\cite{greenwald2003general}.\ma{}

\paragraph{Coarse Correlated Equilibrium}
The \emph{coarse correlated equilibrium} (CCE)~\citep{moulin1978strategically} is a special (unconstrained) Phi-equilibrium whose set of deviations is $\Phi_{\textsc{CCE}} := \{ \Phi_{i,\textsc{CCE}} \}_{i \in \N}$ such that:
\[
	\Phi_{i,\textsc{CCE}} := \left\{ \phi_i \,\Big\vert \,\exists \boldsymbol{h} \in \Delta_{\A_i} : \phi_i[b,a] = \boldsymbol{h}[a] \,\,\, \forall b,a \in \A_i \right\}.
\]
Intuitively, such sets contain all the possible deviations that prescribe the same probability distribution independently of the received action recommendation.
% do \emph{not} depend on the received recommendation.
%
%%
%Consider now the set of deviations defined as:
%\[
%\Phi_{i,\textnormal{CCE}}:=\left\{\phi_i[b,a]=c_a\in[0,1]\,\forall b\in\A: \sum_{a\in\A^i}c_a=1\right\}.
%\]
%This set encodes all possible deviations that do not depend on the recommendation received by the player. When we consider such set for the deviations in absence of costs constraints we recover the definition of \emph{Coarse Correlated Equilibrium}~\cite{greenwald2003general}.

Thus, our constrained Phi-equilibria include the generalization of CEs and CCEs to cost-constrained games.

Our definition of constrained Phi-equilibrium needs to employ ``mixed'' deviations that map action recommendations to probability distributions over actions.
This is necessary in presence of constraints.
Instead, without them, one can simply consider ``pure'' deviations that map recommendations to actions deterministically~\cite{greenwald2003general}.
%
%In our definition of Constrained Phi-equilibrium needs to use the set of ``mixed'' deviations. In literature, the sets are usually comprised of a discrete of ``pure'' deviations. In the absence of constraints considering ``mixed'' or ``pure'' is equivalent, however in the presence of constraints makes it necessary to consider mixed deviations. 

%However we need to consider the convex polytope generated by such deviations since it might be the case that, in the presence of constraints, the hyperplanes of such constraints cut the 

%
%
%\ma{Definire set dei C-CE, C-CCE, ecc}
%\ma{Spiegare perhcè non si possono prendere i $\Phi_{INT}$ della Amy}
%\ma{Spostare qui la definizione del problema 4.1?}

	\section{Challenges of Constrained Phi-equilibria}\label{sec:hard}

In this section, we show that, in cost-constrained normal-form games, Phi-equilibria loose the nice computational properties that they exhibit in unconstrained settings.
This is crucially determined by the fact that the set of constrained Phi-equilibria may \emph{not} be convex in general.
\begin{proposition}\label{prop:non-convex}
	Given any instance $I := (\Gamma, \Phi)$, the set of constrained Phi-equilibria may {not} be convex.
	%	
	%	The set of constrained Phi-equilibria of a cost-constrained normal-form game may {not} be convex.
\end{proposition}
Proposition~\ref{prop:non-convex} is proved by the following example.
\begin{example}\label{example}
	Let $\Phi_{\mathsf{ALL}}$ be the set of all the possible deviations
	%as defined above.
	%
	% This set express all possible deviations of player $i$.
	%
	in a two-player game in which each player has two actions, namely $\A_1=\A_2=\{a_0, a_1\}$.
	%
	% $m=1$, i.e. a single cost constraint $\C_1=\C_2$.
	%
	The first player's utility is such that $u_1(a,a')=0$ for all $a \in \A_1$ and $a' \in \A_2$, while the second player's utility is such that $u_2(a_0,a_1)=1$, and $0$ otherwise.
	Both players share the same single cost constraint ($m=1$).
	Their cost functions are defined as $c_i(a_0,a_1)=1$, $c_i(a_0,a_0)=-\frac{1}{2}$, and $c_i(a_1,a)=-1$ for all $a \in \A_2$.
	Notice that the instance defined above satisfies Assumption~\ref{ass:strictly} for $\rho = 1/2$.
	It is easy to see that the correlated strategy $\zvec^1 \in \Delta_{\A}$ such that $\zvec^1[a_0,a_0]=\frac{2}{3}$ and $\zvec^1[a_0,a_1]=\frac{1}{3}$
	%, and $\zvec^1[a_1, a]=0$  otherwise, 
	is a constrained Phi-equilibrium.
	Moreover, the ``pure'' correlated strategy $\zvec^2 \in \Delta_{\A}$ such that  $\zvec^2[a_1,a_0]=1$ is also a constrained Phi-equilibrium.
	However, the combination $\zvec^3=\frac{1}{2} (\zvec^1+\zvec^2)$ is {not} a constrained Phi-equilibrium.
	Indeed, the second player has an incentive to deviate by using a deviation $\phi_2$ such that $\phi_2[a_0,a_1]=1$ and $\phi_2[a_1,a_1]=1$.
	Such a deviation prescribes to play action $a_1$ when $a_0$ is recommended, and to play action $a_1$ when the recommendation is $a_1$.
	%
	% With this deviation the utility of the second player increases. Indeed, 
	Straightforward calculations show that, for every $a \in \A_1$:
	\[
		(\phi_2\diamond \zvec^3)[a, a^\prime] = 
			\begin{cases}
				\frac{1}{2} &\textnormal{if}\quad a^\prime = a_1\\
				0&\textnormal{otherwise},
			\end{cases}
	\]	
	and $u_2(\phi_2\diamond \zvec^3) = \frac{1}{2}>u_2(\zvec^3)= \frac{1}{6}$.
	Moreover, the deviation is safe, since $\phi_2\in\Phi_2^{\mcS}(\zvec^3)$ as $c_2(\phi_2\diamond \zvec^3)=0$.
\end{example}

	%\section{Hardness Result}\label{sec:hard}

%
%
%Given  a linear function $\ell:\Delta(\A)\to[0,1]$, our goal is to approximate the Constrained Phi-equilibria that maximizes a linear combination of the agents' utility. Formally, we define the following problem.
%
%\ma{dire given an instance in prelim}
%\ma{definire i CPE?}
%\begin{definition}
%	Given a set of deviations for each player $\Phi:=\{\Phi_i\}_{i\in[N]}$, a linear function $\ell$, and two values $\alpha, \epsilon\in [0,1]$, \textsc{ApxCPE}$(\alpha,\epsilon)$ is the problem of finding a Costrained $\epsilon$-Phi-equilibrium $\zvec$ such that $\ell(\zvec)\ge \alpha\ell(\zvec^\star)$, where $\zvec^*$ is the optimal Constrained Phi-equilibrium for the linear function $\ell$.
%	%\ma{definire $OPT$... Per ora ho lasciato $L_{\D,0}$.}
%\end{definition}

In order to formally asses the computational challenges of computing constrained Phi-equilibria, we prove the following strong inapproximability result:
\begin{restatable}[Hardness]{theorem}{theoremHardness}\label{thm:hardness}
	For any constant $\alpha>0$, the problem \textsc{ApxCPE}$(\alpha,(\alpha/s)^2)$ is \NPHard, where $s$ is the number of players' actions in the instance given as input.
	%	
	%	\textcolor{red}{For any $\alpha>0$, \textsc{ApxCPE}$(\alpha,{\alpha^2}/{|\A|^2})$ is \NPHard\ for any $\alpha>0$ and $0 < \epsilon \leq {\alpha^2}/{|\A|^2}$, where $|\A|$ is the number of ???? of the instance given as input.}
	%
	%even if $\ell$ is the linear function defining the social welfare.
\end{restatable}
Intuitively, Theorem~\ref{thm:hardness} states that, for every multiplicative approximation factor $\alpha > 0$, it is \emph{not} possible to find a constrained $\epsilon$-Phi-equilibrium having value of $\ell$ at least a fraction $\alpha$ of its optimal value in time polynomial in $\frac{1}{\epsilon}$.
Moreover, as a byproduct of Theorem~\ref{thm:hardness}, we also get the inapproximability up to within any factor of the problem of computing an optimal constrained (exact) Phi-equilibrium. 
%
%\textcolor{red}{As a corollary of Theorem~\ref{thm:hardness}, we also get the inapproximability up to within any factor of the problem of computing an constrained Phi-equilibrium which approximates a linear function $\ell$ of an optimal one.}

%\begin{remark}
%
Notice that the hardness result in Theorem~\ref{thm:hardness} cannot hold for values of $\epsilon$ that are \emph{independent} from the instance size.
Indeed, as we prove in Corollary~\ref{cor:res_2} in Section~\ref{sec:optimal}, problem \textsc{ApxCPE}$(1,\epsilon)$ can be solved in quasi-polynomial time in the instance size whenever $\epsilon > 0$ is a given constant.
Thus, any \NPHard ness result for \textsc{ApxCPE}$(\alpha,\epsilon)$ would contradict the \emph{exponential-time hypothesis}.\footnote{The exponential-time hypothesis conjectures that solving 3SAT requires at least exponential time.}
%
%\textcolor{red}{Notice that, in order to prove the hardness result in Theorem~\ref{thm:hardness}, it is necessary to use an $\epsilon$ depending on the instance size.
%%
%Indeed, as we prove in the following sections, whenever $\epsilon$ is a constant the problem can be solved in quasi-polynomial time.
%%
%Hence, an hardness result for a constant $\epsilon$ would contradict the exponential-time hypothesis~\citep{boh}.}
% assuming the Exponential-time hypothesis, the problem with constant $\epsilon$ cannot be\NPHard. 
%
%\end{remark}
	\section{Computing Optimal Constrained $\epsilon$-Phi-equilibria Efficiently}\label{sec:optimal}

In this section, we show how to circumvent the negative result established by Theorem~\ref{thm:hardness}.
In particular, we prove that, when the number of cost constraints is fixed, problem \textsc{ApxCPE}$(1,\epsilon)$ can be solved in time polynomial in the instance size and $\frac{1}{\epsilon}$ for $\epsilon > 0$ (Corollary~\ref{cor:res_1}).
Moreover, we also prove that, in general, for any constant $\epsilon > 0$ problem \textsc{ApxCPE}$(1,\epsilon)$ admits a quasi-polynomial-time algorithm, whose running time becomes polynomial when the number of players' actions is fixed (Corollary~\ref{cor:res_2}).
%
%, whenever the number of players and that of actions are fixed, the same problem admits a quasi-polynomial-time algorithm (Corollary~\ref{cor:res_2}).

First, in Section~\ref{sec:optimal_1}, we provide a general algorithm that is at the core of the two main results of this section.
Then, in Section~\ref{sec:optimal_1}, we show how the algorithm can be suitably instantiated in order to prove each result.
In the rest of this section, unless stated otherwise, we always assume that an $\epsilon > 0$ has been fixed, and that $I := (\Gamma,\Phi)$ and $\ell : \Delta_{\A} \to \mathbb{R}$ are the inputs of a given instance of problem \textsc{ApxCPE}$(1,\epsilon)$.

\subsection{General Algorithm}\label{sec:optimal_1}

The main technical tool that we employ in order to design our algorithm is a ``Lagrangification'' of the constraints defining the sets $\Phi_i^{\mcS}(\zvec)$ of safe deviations.
%
% By means of duality arguments, this allows us to reduce our problem to that of finding optimal Phi-equilibria in games without cost constraints.
%
%\ma{recap iniziale}
%\ma{dire che fissiamo un Phi qualsiasi che soddisfa ass 1}
%In this section, we show that whenever Assumption~\ref{ass:strictly} is satisfied and the number of constraints is contant, the problem  ???-$(\ell,1,\epsilon)$ can be solved in polynomial-time in the instance size and $\frac{1}{\epsilon}$. anche QPTAS/PTAS quando num actions constante.
%
%As a first step, we show that it is possible to Lagrangify and linearize the constraint regarding the safety of the violations. Formally, we prove the following.
%
%
First, we prove the following preliminary result, which shows that strong duality holds for the problem $\max_{\phi_i \in \Phi_i^\mcS(\zvec)} u_i(\phi_i \diamond \zvec)$ of finding the best safe deviation for player $i \in \N$ at $\zvec \in \Delta_{\A}$.
\begin{restatable}{lemma}{minmax}\label{lem:minmax}
	For every $\zvec\in \Delta_{\A}$ and $i \in \N$, it holds 
	\begin{align*}
		\max_{\phi_i \in \Phi_i^\mcS(\zvec)} u_i(\phi_i \diamond \zvec)& =\sup\limits_{\phi_i\in\Phi_i} \inf_{\etavec_i \in \mathbb{R}^{m}_+}  \left(u_i(\phi_i\diamond \zvec)-\etavec_i^{\top} \cvec_i(\phi_i \diamond \zvec)\right) \\
		& =\inf_{\etavec_i \in \mathbb{R}^{m}_+}  \sup\limits_{\phi_i\in\Phi_i} \left(u_i(\phi_i \diamond\zvec)-\etavec_i^{\top} \cvec_i(\phi_i \diamond\zvec)\right).
	\end{align*}
\end{restatable}
Then, by exploiting Lemma~\ref{lem:minmax}, we can prove that, under Assumption~\ref{ass:strictly}, strong duality continues to hold even when restricting the Lagrange multipliers $\etavec_i $ to have $\ell_1$-norm less than or equal to $1/\rho$.
%
%Let $V_{z,i}=\min_{\eta_i \in \mathbb{R}_+^?} \max\limits_{\phi^i}\left(U^i(\phi^i z)-\eta^i C(\phi^i z)\right)= \max\limits_{\phi^i} \min_{\eta_i \in \mathbb{R}_+^?} \left(U^i(\phi^i z)-\eta^i C(\phi^i z)\right)  $
%
Formally:
\begin{restatable}{lemma}{duality}\label{lem:duality}
	Let $\D\coloneqq \left\{\etavec\in\Reals^{m}_+ \mid  ||\etavec||_1\le{1}/{\rho}\right\}$.
	Then, for every $\zvec\in \Delta_{\A}$ and $i \in \N$, it holds:
	\begin{align*}
	\max_{\phi_i \in \Phi_i^{\mcS}(\zvec)} u_i(\phi_i \diamond \zvec)	& =  \max\limits_{\phi_i\in\Phi_i} \min_{\etavec_i \in  \D} \left(u_i(\phi_i \diamond \zvec)-\etavec_i^\top \cvec_i(\phi_i\diamond \zvec)\right) \\
	&=\min_{\etavec_i \in \D} 	\max\limits_{\phi_i\in\Phi_i}\left(u_i(\phi_i \diamond \zvec)-\etavec_{i}^\top \cvec_i(\phi_i \diamond \zvec)\right) .
	 \end{align*}
\end{restatable}
Lemma~\ref{lem:duality} allows us to write player $i$'s incentive constraints in the definition of constrained $\epsilon$-Phi-equilibria as
\begin{equation}\label{eq:cons_lemmiz}
	u_i(\zvec) \geq \min_{\etavec_i \in \D} 	\max_{\phi_i\in\Phi_i}\left(u_i(\phi_i \diamond \zvec)-\etavec_{i}^\top \hspace{-0.1mm} \cvec_i(\phi_i \diamond \zvec)\right) - \epsilon.
\end{equation}
%
% $u_i(\zvec) \geq \min_{\etavec_i \in \D} 	\max_{\phi_i\in\Phi_i}\left(u_i(\phi_i \diamond \zvec)-\etavec_{i}^\top \cvec_i(\phi_i \diamond \zvec)\right) - \epsilon$.

This crucially allows us to show the following result: solving problem \textsc{ApxCPE}$(1,\epsilon)$ is equivalent to computing $\max_{(\etavec_1,\ldots,\etavec_n)\in\D^n} F_{\epsilon} (\etavec_1, \ldots, \etavec_{n})$, where $F_{\epsilon} (\etavec_1, \ldots, \etavec_{n})$ is the optimal value of a suitable maximization problem parameterized by tuples of Lagrange multipliers $\etavec_i \in \D$, one per player $i \in \N$.
%
% In order to do so, we first introduce a maximization problem parameterized by a tuple of Lagrange multipliers $\etavec_i \in \mathbb{R}_+^m$, one for each player $i \in \N$.
%
Such a problem
%, which we call $F_{\ell,\epsilon}(\etavec_1, \ldots, \etavec_{n})$, 
asks to compute a safe correlated strategy maximizing the linear function $\ell$ subject to players' incentive constraints that are re-formulated by means of Lemma~\ref{lem:duality}.
%
%that are linearized and relaxed by an additive $\epsilon$ amount.
%
Formally, we define $F_{\epsilon}(\etavec_1, \ldots, \etavec_{n})$ as the maximum of $\ell(\zvec)$ over those $\zvec \in \mcS$ that additionally satisfy the following constraint for every $i \in \N$:
\begin{equation}\label{eq:LP2}
	u_i(\zvec)\ge \max\limits_{\phi_i\in\Phi_i}\left(u_i(\phi_i\diamond \zvec)-\etavec_i^\top \cvec_i(\phi_i \diamond \zvec)\right) -\epsilon .
\end{equation}
Notice that the $\min$ operator that appears in the right-hand side of Constraints~\eqref{eq:cons_lemmiz} is dropped by adding the outer maximization over the tuples $(\etavec_1,\dots,\etavec_n) \in \D^n$, as the maximum of $\ell$ is always achieved when the right-hand side of such constraints is as small as possible.

Next, we show that $F_{\epsilon}(\etavec_1, \ldots, \etavec_{n})$ can be computed in polynomial time by means of the \emph{ellipsoid algorithm}.
\begin{lemma}\label{lem:efficientF}
	For every tuple $(\etavec_1,\dots,\etavec_n) \in \D^n$, the value of $F_{\epsilon}(\etavec_1,\dots,\etavec_n)$ can be computed in time polynomial in the instance size $|I|$ and $\frac{1}{\epsilon}$.
\end{lemma}

\begin{proof}
	We show that $F_{\epsilon}(\etavec_1,\ldots,\etavec_n)$ can be solved in polynomial time by means of the {ellipsoid algorithm}.
	Let us notice that Constraints~\eqref{eq:LP2} can be equivalently encoded by a set of linear inequalities, one for each player $i\in \N$ and deviation $\phi_i\in vert(\Phi_i)$, where $vert(\Phi_i)$ denotes the set of vertexes of the polytope $\Phi_{i}$ (recall Assumption~\ref{ass:solvablePolytope}).
	Thus, solving $F_{\epsilon}(\etavec_1,\ldots,\etavec_n)$ is equivalent to solving an LP with a (possibly) exponential number of constraints, but polynomially-many variables.
	Such an LP can be solved in polynomial time by means of the ellipsoid algorithm, provided that a polynomial-time separation oracle for the linearized version of Constraints~\eqref{eq:LP2} is available.
	Such an oracle can be implemented by solving the maximization in the right-hand side of Constraints~\eqref{eq:LP2} for a correlated strategy $\zvec\in\Delta_{\A}$ given as input.
	Formally, the separation oracle solves the following problem for each player $i \in \N$:
	\[
		\phi_i^\star\in\arg\max\limits_{\phi_i\in\Phi_i}\left(u_i(\phi_i\diamond\zvec)-\etavec_i^\top\cvec_i(\phi_i\diamond\zvec)\right),
	\]
	which can be done efficiently thanks to Assumption~\ref{ass:solvablePolytope}.
	Then, if the separation oracle finds any $\phi_i^\star$ such that:
	\[
	u_i(\zvec)\ge u_i(\phi^\star_i\diamond\zvec)-\etavec_i^\top\cvec_i(\phi^\star_i\diamond\zvec),
	\]
	it outputs the above inequality as a separating hyperplane to be used in the ellipsoid algorithm.
\end{proof}
%
%
%We can use the previous results of Lemma~\ref{lem:duality} states that a solution to $\max_{(\etavec_1,\ldots,\etavec_N)\in\D^N} F_{\ell,\epsilon} (\etavec_1, \ldots, \etavec_{N})$ would yield an optimal Constrained $\epsilon$-Phi-equilibrium. This is readily seen as the variables $\etavec_i\in\D$ only appear in the constraints of the problem $F_{\ell,\epsilon} (\etavec_1, \ldots, \etavec_{N})$ and the maximum solution is achieved when its constraints are the biggest possible, which happens when the right hand side of Equation~\ref{eq:LP2} is minimized.\ma{rewording}\alb{non si capisce}

Lemma~\ref{lem:efficientF} is \emph{not} enough to complete our algorithm, since we need an efficient way of optimizing $F_\epsilon (\etavec_1,\ldots,\etavec_n)$ over all the tuples of Lagrange multipliers.
This problem is non-trivial, since $F_\epsilon (\etavec_1,\ldots,\etavec_n)$ is non-concave in $\etavec_i$.
Nevertheless, we show that, by restricting the domain $\D$ of the Lagrange multipliers to a suitably-defined finite ``small'' subset, we can still find a constrained $\epsilon$-Phi-equilibrium whose value of $\ell$ is at least as large as that of any constrained (exact) Phi-equilibrium.
This is enough to solve \textsc{ApxCPE}$(1,\epsilon)$.
%s
In particular, we need a finite subset of ``good'' Lagrange multipliers, in the sense of the following definition.
%
%However, this problem is non-concave in $\etavec_i$ and cannot be approximated efficiently (this follows from Theorem~\ref{thm:hardness}).
%Nonetheless, we show that restricting to finite sufficiently small sets of Lagrangian multipliers is sufficient to compute optimal approximate Constrained Phi-equilibrium.
%In particular, we define a finite set of ``good'' Lagrangian multipliers that enjoy the following property.
%
\begin{definition}
	Given any $\delta > 0$, a set $\tilde \D \subseteq \D$ is \emph{$\delta$-optimal} if, for every $\zvec \in \Delta_{\A}$ and $i \in \N$, the following holds:
	\begin{align*}
		\min\limits_{\etavec_i \in \tilde \D}\max\limits_{\phi_i\in\Phi_i}&\left(u_i(\phi_i\diamond \zvec)-\etavec_i^\top \cvec_i(\phi_i \diamond \zvec)\right) \leq \max\limits_{\phi_i \in \Phi^\mcS_i(\zvec)} u_i(\phi_i\diamond\zvec)+\delta.
	\end{align*}
	%	
	%	A set $\bar \D\subseteq \D$ is $\delta$-optimal if for each $\zvec \in \Delta(\A)$ and for each $i\in[N]$ we have
	%	\begin{align}
	%		\min\limits_{\etavec_i \in \bar\D}\max\limits_{\phi_i\in\Phi_i}&\left(U_i(\phi_i\diamond \zvec)-\etavec_i^\top \Cvec_i(\phi_i \diamond \zvec)\right)\\ 
	%		\le &\max\limits_{\phi_i \in \Phi^\mcS_i(\zvec)} U_i(\phi_i\diamond\zvec)+\delta.
	%	\end{align}
\end{definition}
Intuitively, thanks to Lemma~\ref{lem:duality}, if we optimize the Lagrange multipliers over a $\delta$-optimal set $\tilde \D \subseteq \D$, instead of optimizing them over $\D$, then we are allowing the players to violate incentive constraints by at most $\delta$.

In the following, we assume that a finite $\delta$-optimal set $\tilde \D \subseteq \D$ is available. 
In Section~\ref{sec:Dsets}, se show how to design two particular $\delta$-optimal sets that allow to prove our main results.
For ease of presentation, we let 
\[
	L_{\tilde \D, \epsilon}:=\max\limits_{(\etavec_1, \ldots, \etavec_n )\in  \tilde \D^n} F_{\epsilon}(\etavec_1, \ldots, \etavec_n)
\]
be the optimal value of $F_\epsilon (\etavec_1,\ldots,\etavec_n)$ when the Lagrange multipliers are constrained to be in a $\delta$-optimal set $\tilde \D \subseteq \D$.  
%
%Given a set $\bar \D \subseteq \D$, consider the problem:\ma{non definire $L$}\ma{mettere $\epsilon$ a pedice di $F$}
%\[
%L_{\bar\D, \epsilon}:=\max\limits_{(\etavec_1, \ldots, \etavec_N )\in \bar \D^N} F_{\ell,\epsilon}(\etavec_1, \ldots, \etavec_N).
%\]
%
% We observe that $\epsilon$ appears in the problem defining $L_{\D_\delta,\epsilon}$ only in the constraints of $F_{\epsilon}$.
%
% Moreover, a larger value of $\epsilon$ only makes such constraints looser, yielding higher values of $\ell$.
%
Next, we show that, given any $\delta$-optimal set $\tilde \D$ with $\delta \le \epsilon$, the value of $L_{\tilde \D,\epsilon}$ is at least that achieved by constrained (exact) Phi-equilibria, namely $L_{\D,0}$.
Formally:
%
%the exact solution (\ie the one with $\epsilon=0$ and the sets of multipliers being $\D$ which is $0$-optimal by Lemma~\ref{lem:duality}).
%
\begin{restatable}{lemma}{deltaOPTProblem}\label{lem:deltaOPT}
	Given any $0 < \delta \leq \epsilon$ and a $\delta$-optimal set $\tilde \D \subseteq \D$, the following holds: $L_{\tilde \D, \epsilon}\ge L_{\D, 0}$.
	%
	% For any $\delta\ge\epsilon$\ma{TBC}, and $\bar\D$ which is $\delta$-optimal we have that $L_{\bar\D, \epsilon}\ge L_{\D, 0}$.
\end{restatable}
Intuitively, Lemma~\ref{lem:deltaOPT} is proved by noticing that, provided that $\delta \leq \epsilon$, the incentive constraints violation introduced by using $\tilde \D$ instead of $\D$ is at most $\epsilon$.
Moreover, the set of feasible correlated strategies can only expand by allowing incentive constraints to be violated, and, thus, the value of the objective $\ell$ can only increase.

Lemma~\ref{lem:deltaOPT} suggests a way of solving \textsc{ApxCPE}$(1,\epsilon)$.
%
% Lemma~\ref{lem:deltaOPT} suggests a way of finding an optimal constrained $\epsilon$-Phi-equilibrium.
%
Indeed, given a finite $\delta$-optimal set $\tilde \D \subseteq \D$ with $\delta \leq \epsilon$, by enumerating over all the tuples of Lagrange multipliers $\etavec_i\in \tilde \D$, one per player $i \in \N$, we can find the desired constrained $\epsilon$-Phi-equilibrium.
The following theorem shows that this procedure gives an algorithm for \textsc{ApxCPE}$(1,\epsilon)$ that runs in time polynomial in the instance size, $|\tilde \D|$, and $\frac{1}{\epsilon}$.
%, when the number of players is fixed.
%
%there exists a polynomial algorithm (\ie polynomial in $1/\epsilon$) whenever the number of constraints and number of player is constant.
%
\begin{theorem}\label{th:polyalgo}
	Given a finite $\delta$-optimal set $\tilde \D \subseteq \D$ with $\delta \leq \epsilon$, there exists an algorithm that solves \textsc{ApxCPE}$(1,\epsilon)$ and runs in time polynomial in the instance size $ |I| $, the number $|\tilde \D|$ of elements in $\tilde \D$, and $\frac{1}{\epsilon}$ for every $\epsilon > 0$.
	%	
	%	For any $\epsilon > 0$, given a $\delta$-optimal set $\D_\delta$ with $\delta \leq \epsilon$, there exists an algorithm that solves \textsc{ApxCPE}$(1,\epsilon)$ and runs in time polynomial in the instance size $|I|$ and $\frac{1}{\epsilon}$ when the number of players $n$ is fixed.
	%	
	%	Let $\bar \D$ be a discrete $\delta$-optimal set. Then, there exists a $\poly(|\bar\D|,|I|)$ algorithm that find an $\delta$-SCE with at least the optimal social welfare, where $|I|$ is the size of the instance.\ma{allungare lo statement}
\end{theorem}

\begin{proof}
	The algorithm works by enumerating over all the possible tuples of Lagrange multipliers $\etavec_i\in \tilde \D$, one per player $ i \in \N$.
	These are polynomially many in the size $|\tilde \D|$ when the number of players $n$ is fixed. 
	For every tuple $(\etavec_1,\ldots,\etavec_n) \in \tilde \D^n$, the algorithm solves $F_{\epsilon}(\etavec_1,\ldots,\etavec_n)$, which can be done in time polynomial in $|I|$ and $\frac{1}{\epsilon}$ thanks to Lemma~\ref{lem:efficientF}.
	Finally, the algorithm returns the correlated strategy $\zvec \in \Delta_{\A}$ with the highest value of $\ell$ among those computed while solving $F_{\epsilon}(\etavec_1,\ldots,\etavec_n)$.
	It is easy to see that the returned solution solves problem \textsc{ApxCPE}$(1,\epsilon)$ by applying Lemma~\ref{lem:deltaOPT}.
	This concludes the proof.
\end{proof}

\subsection{Instantiating the General Algorithm}\label{sec:Dsets}

Next, we show how to build $\delta$-optimal sets $\tilde \D$ that, when they are plugged in the algorithm in Theorem~\ref{th:polyalgo}, allow us to derive our results.
In particular, we consider the set:
\[
	\D_\tau := \left\{ \etavec \in\D \, \Big\vert \,  \eta_{j} = k \tau, k \in \left\{0,\dots,\lfloor\nicefrac{1}{\tau \rho} \rfloor \right\}  \, \forall j \in [m] \right\},
\]
which is a discretization of $\D$ with a regular lattice of step $\tau \in \mathbb{R}_+$ (notice that $\eta_j$ is the $j$-th component of $\etavec$).
%
%\[
%	\D_\epsilon=\left\{\etavec_i\in\D: \eta_{i,j}=k\epsilon,\,k \in \{0,\dots,\lfloor\nicefrac{1}{\epsilon \rho} \rfloor\},\, \forall j\in[m]\right\},
%\]
%\ie we discretize $\D$ with a regular lattice of size $\epsilon$.
%
By a simple stars and bars combinatorial argument, we have that $|\D_\tau|=\binom{\lfloor\nicefrac{1}{\tau \rho} \rfloor+m}{m}$.
Thus, since it holds that $|\D_\tau|=O(\left(\nicefrac{1}{\tau \rho}\right)^{m})$, if the number of constraints $m$ is fixed, $|\D_\tau|$ is bounded by a polynomial in $\nicefrac{1}{\tau \rho}$.
% whenever $m=O(1)$.
%
Moreover, simple combinatorial arguments show that $|\D_\tau| \le (1+m)^{\lfloor\nicefrac{1}{\tau \rho}\rfloor}$.\footnote{See Appendix~\ref{sec:app_optimal} for a formal proof.}
%
% we can use the well known bound $k!\ge (k/3)^k$ to write the following upperbound $\binom{\lfloor\nicefrac{1}{\epsilon \rho} \rfloor+m}{m}\le \left(3\frac{\lfloor\nicefrac{1}{\epsilon \rho} \rfloor+m}{\lfloor\nicefrac{1}{\epsilon \rho}\rfloor}\right)^{\lfloor\nicefrac{1}{\epsilon \rho} \rfloor}\le (1+m)^{\lfloor\nicefrac{1}{\epsilon \rho}\rfloor}$. 
%
Thus, it also holds that $|\D_\tau| = O(m^{\nicefrac{1}{\tau \rho} })$.
Notice that the two bounds on $|\D_\tau|$ are non-comparable, and, thus, they give rise to two distinct results, as we show in the following.
%
%
%In the following, we will use the most appropriate one depending on specific assumptions, later ot be discussed.
%
%\ma{Notice that whenever there are a constant number of constraints per player $\D_\epsilon$ has polynomial-size $O(\frac{1}{\epsilon}^{\#constraints N})$. Moreover, the size of the set is upperbounded by $O({\#constraints N}^{\frac{1}{\epsilon}}$.???}

By using the first bound on $|\D_\tau|$, we can show that the set $\D_\tau$ is $\delta$-optimal for $\delta = m \tau$.
Formally:
\begin{restatable}{lemma}{gridisemopt}\label{lem:bound_1}
	For any $\tau>0$, the set $\D_\tau$ is $\left(\tau m\right)$-optimal.
\end{restatable}
Thus, whenever the number $m$ of cost constraints is fixed, Lemma~\ref{lem:bound_1}, together with Theorem~\ref{th:polyalgo}, allows us to provide a polynomial-time algorithm. 
%
% for problem \textsc{ApxCPE}$(1,\epsilon)$ that runs in time polynomial in the instance size and $\frac{1}{\epsilon}$ for every $\epsilon > 0$ when the number $m$ of cost constraints is fixed.
%
Indeed, it is sufficient to apply Theorem~\ref{th:polyalgo} for the $(\tau m)$-optimal set $\D_\tau$ with $\tau := \epsilon / m$ to obtain the following first main result:
\begin{corollary}\label{cor:res_1}
	There exists an algorithm that solves problem \textsc{ApxCPE}$(1,\epsilon)$ in time polynomial in $|I|$ and $\frac{1}{\epsilon}$ for every $\epsilon > 0$, when the number $m$ of cost constraints is fixed.
	%	
	%	FPTAS when number of constraint is constant.
	%	In particular, for each $\delta$, setting $\epsilon=\delta/|C|$ and taking $\bar \D=\D_\epsilon$, we obtain a $poly(\D_\epsilon,I)$ algorithm, \emph{i.e.}, $poly( (\delta/ |C|)^{|C|},I)$.
\end{corollary}

On the other hand, by using the second bound on $|\D_\tau|$, we can show that $\D_\tau$ is $\delta$-optimal for $\delta$ depending logarithmically on the number of players' actions.
Formally:
\begin{restatable}{lemma}{uniform}\label{lem:bound_2}
	For any $\tau>0$, the set $ \D_\tau$ is $\delta$-optimal for $\delta =2\sqrt{{2\tau}\log s / \rho}$, where $s$ is the number of players' actions.
\end{restatable}
Lemma~\ref{lem:bound_2} (together with Theorem~\ref{th:polyalgo}) immediately gives us a quasi-polynomial-time for solving \textsc{ApxCPE}$(1,\epsilon)$ for a given constant $\epsilon > 0$.
Moreover, its running time becomes polynomial when the number of players' actions is fixed.
\begin{corollary}\label{cor:res_2}
	For any constant $\epsilon > 0$, there exists an algorithm that solves \textsc{ApxCPE}$(1,\epsilon)$ in time $O(|I|^{\log s})$.
	Moreover, when the number $s$ of players' actions is fixed, the algorithm runs in time polynomial in $|I|$.
	%
	% In particular, for each $\delta$, setting $k=\log(|A|)/\epsilon^2$ and taking $\bar \D=\hat \D_k$, we obtain a $poly(\hat \D_k,I)$ algorithm, \emph{i.e.}, $poly( |C|^k, I)$. QPTAS. When the number of actions is constant,
	% PTAS.
\end{corollary}
Notice that it is in general \emph{not} possible to design an algorithm that runs in time polynomial in $\frac{1}{\epsilon}$, since this would contradict the hardness result in Theorem~\ref{thm:hardness}.

	\section{A Special Case: Deviation-dependent Costs}\label{sec:easy}

We complete our computational study of constrained Phi-equilibria by considering a special case in which player $i$'s costs associated to a deviation $\phi_i$ only depend on $\phi_i$ and \emph{not} on the (overall) modified correlated strategy $\phi_i\diamond\zvec$.

We consider instances satisfying the following assumption.
\begin{assumption}\label{ass:cce}
	For every player $i\in \N$ and player $i$'s deviation $\phi_i\in\Phi_i$, there exists a function $\tilde\cvec_i: \Phi_i\to[-1,1]^m$ such that $\tilde\cvec_i(\phi_i) :=\cvec_i(\phi_i \diamond\zvec)$ for every $\zvec \in \Delta_{\A}$.
	%
	% it holds that there exists $\tilde\Cvec_i:\Phi_i\to[-1,1]^m$ such that $\tilde\Cvec_i(\phi_i)=\Cvec_i(\phi_i \diamond\zvec)$.
\end{assumption}
%
%Thus, with abuse of notation, we can write $\Cvec_i(\phi_i\diamond\zvec)\equiv\tilde\Cvec_i(\phi_i)$ whenever Assumption~\ref{ass:cce} holds.
%
Notice that, whenever Assumption~\ref{ass:cce} holds, the set $\Phi_i^\mcS(\zvec)$ of safe deviations does \emph{not} depend on $\zvec$.
Thus, in the rest of this section, we write w.l.o.g. $\Phi_i^\mcS$ rather than $\Phi_i^\mcS(\zvec)$.
%
% We remark that Assumption~\ref{ass:cce} is in general a joint property of the set $\Phi_i$ and the costs $\Cvec_i$.

A positive effect of Assumption~\ref{ass:cce} is that it recovers the convexity of the set of constrained Phi-equilibria, rendering them more akin to unconstrained ones.
Formally:
\begin{restatable}{proposition}{convex}\label{thm:convex}
	For instances $I := (\Gamma,\Phi)$ satisfying Assumption~\ref{ass:cce}, the set of constrained $\epsilon$-Phi-equilibria is convex.
	% whenever Assumption~\ref{ass:cce} holds.
	%
	% The set of constrained $\epsilon$-Phi-equilibria is convex whenever Assumption~\ref{ass:cce} holds.
\end{restatable}
Proposition~\ref{thm:convex} suggests that constrained Phi-equilibria are much more computationally appealing under Assumption~\ref{ass:cce} than in general, as we indeed show in the rest of this section.

% After formally introducing the model studied in this section, in 
%
First, in Section~\ref{sec:easy_optimal}, we show that \textsc{ApxCPE}$(1,0)$ admits a polynomial-time algorithm under Assumption~\ref{ass:cce}.
Then, in Section~\ref{sec:easy_learning}, we design a no-regret learning algorithm that efficiently computes \emph{one} constrained $\epsilon$-Phi equilibrium with $\epsilon= O(\nicefrac{1}{\sqrt{T}})$ as the number of rounds $T$ grows.
Finally, in Section~\ref{sec:easy_cce}, we provide a natural example of constrained Phi-equilibria satisfying Assumption~\ref{ass:cce}.
%
%
%
%We complete the study of Constrained Phi-equilbria by considering a special case, in wich the costs only depend on the deviation $\phi_i$ chosen by the $i$-th player and not by the ovarall modified correlated strategy $\phi_i\diamond\zvec$. At the end of the section\ma{o meglio subito?} we will give a natural example that satisfies such hypothesis.

\subsection{A Poly-time Algorithm for Optimal Equilibria}\label{sec:easy_optimal}

We prove that, whenever Assumption~\ref{ass:cce} holds, the problem of computing an (exact) Phi-equilibrium maximizing a given linear function can be solved in polynomial time.
This is done by formulating the problem as an LP with polynomially-many variables and exponentially-many constraints, which can be solved by means of the ellipsoid method, similarly to how we compute $F_\epsilon(\etavec_1,\ldots\etavec_n)$ in Section~\ref{sec:optimal} (see the proof of Lemma~\ref{lem:efficientF}). 
%
% The problem of finding an optimal Constrained Phi-equilibrium can readily written as an LP which we can solve with the ellipsoid algorithm, similarly to Lemma~\ref{lem:efficientF} we can state the following:
%
Formally:
\begin{restatable}{theorem}{optimalsimple}\label{th:optimal_MCCCE}
	Restricted to instances $I := (\Gamma,\Phi)$ which satisfy Assumption~\ref{ass:cce}, \textsc{ApxCPE}$(1,0)$ admits a polynomial-time algorithm.
	%
	%then we can find an exact constrained Phi-equilibrium in polynomial time.
\end{restatable}

\subsection{An Efficient No-regret Learning Algorithm}\label{sec:easy_learning}

Next, we show how Assumption~\ref{ass:cce} allows us to find a constrained $\epsilon$-Phi-equilibrium by means of a polynomial-time decentralized no-regret learning algorithm.
Our algorithm is based on the Phi-regret minimization framework introduced by~\citet{greenwald2003general}, which needs to be extended in order to be able to work with polytopal sets $\Phi_i^\mcS$ of safe deviations, rather than finite sets of ``pure'' deviations.
%
% In this section we will show how to find approximate Constrained Phi-equilibria in a simple decentralized manner when Assumption~\ref{ass:cce} holds.
% The techniques used here are related to the no Phi-regret learning literature.

\begin{algorithm}[!htp]\caption{{Learning a Constrained $\epsilon$-Phi-equilibria}}\label{alg:noregret}
	\begin{algorithmic}[1]
		\REQUIRE Regret minimizers $\mfR_i$ for the sets $\Phi_i^\mcS$, for $i\in \N$
		\STATE Initialize the regret minimizers $\mfR_i$
		\FOR{$t=1, \ldots, T$}
		\FOR{each player $i \in \N$}
		\STATE $\phi_{i,t}\gets \mfR_i.\textsc{Recommend}()$\label{lin:recc}
		\STATE Play according to a distribution $\xvec_{i,t}\in\Delta_{\A_i}$ s.t. $$\xvec_{i,t}[a] = \sum_{b\in\A_i}\phi_{i,t}[b,a]\xvec_{i,t}[b] \quad \forall a\in\A_i$$
		\ENDFOR
		\STATE $\zvec_t\gets\otimes_{i\in \N} \, \xvec_{i,t}$
		\STATE $\mfR_i.\textsc{Observe}( \phi_i \mapsto u_i(\phi_i \diamond \zvec_t) )$\label{lin:obs}
		\ENDFOR
		\STATE \textbf{return} $\bar\zvec_T := \frac{1}{T}\sum_{t=1}^T\zvec_t$
	\end{algorithmic}
\end{algorithm}

Algorithm~\ref{alg:noregret} outlines our no-regret algorithm.
It instantiates a regret minimizer $\mfR_i$ for the polytope $\Phi_i^\mcS$ for each $i\in \N$.
$\mfR_i$ is an object that, at each round $t \in [T]$, recommends a safe deviation $\phi_{i,t} \in \Phi_i^\mcS$ to player $i$ (Line~4 of Algorithm~\ref{alg:noregret}), and, then, observes a function $\phi_i \mapsto u_i(\phi_i \diamond \zvec_t)$ that specifies the utility that would have been obtained by selecting any safe deviation $\phi_i \in \Phi_i^\mcS$ at round $t$  (Line~8 of Algorithm~\ref{alg:noregret}).
$\mfR_i$ guarantees that the regret $R_i^T$ cumulated by player $i$ over $[T]$ grows sublinearly, \ie $R_i^T = o(T)$, where:
\[
	R_i^T := \max_{\phi_i \in \Phi_i} \sum_{t=1}^T  u_i(\phi_i\diamond\zvec_t)-\sum_{t=1}^T u_i(\phi_{i,t}\diamond\zvec_t),
\]
which is how much player $i$ loses by selecting $\phi_{i,t}$ at each $t$ rather than choosing the same best-in-hindsight deviation at all rounds.
Notice that, by taking inspiration from the Phi-regret framework~\citep{greenwald2003general}, given a recommended deviation $\phi_{i,t}$, player $i$ actually plays according to a probability distribution $\xvec_{i,t} \in \Delta_{\A_i}$, which is a stationary distribution of the matrix representing $\phi_{i,t}$.
This is crucial in order to implement the algorithm in a decentralized fashion and to provide convergence guarantees to constrained $\epsilon$-Phi-equilibria (see Theorem~\ref{th:learning_MCCCE}). 
All the distributions $\xvec_{i,t}$ jointly determine a correlated strategy $\zvec_t \in \Delta_{\A}$ at each round $t \in [T]$, defined as $\zvec_t := \otimes_{i\in \N} \, \xvec_{i,t}$, where $\otimes$ denotes the product among distributions; formally, $\zvec_t[\avec] := \prod_{i \in \N} \xvec_{i,t}[a_i]$ for all $\avec \in \A$.

Algorithm~\ref{alg:noregret} provides the following guarantees:
\begin{restatable}{theorem}{learning}\label{th:learning_MCCCE}
	Given an instance $I := (\Gamma,\Phi)$ satisfying Assumption~\ref{ass:cce}, after $T \in \mathbb{N}_{>0}$ rounds, Algorithm~\ref{alg:noregret} returns a correlated strategy $\bar\zvec_T \in \Delta_{\A}$ that is a constrained $\epsilon_T$-Phi-equilibrium with $\epsilon_T = O(\nicefrac{1}{\sqrt{T}})$.
	Moreover, each round of Algorithm~\ref{alg:noregret} runs in polynomial time.
	%
	% Algorithm~\ref{alg:noregret}, after $T$ iterations, returns a Constrained $(\epsilon_T)$-equilibrium, where $\epsilon_T=O(T)$. Moreover the algorithm runs in polynomial time.
\end{restatable}
Let us remark that the crucial property which allows us to design Algorithm~\ref{alg:noregret} is that the sets $\Phi_i^\mcS$ of safe deviations do \emph{not} depend on players other than $i$.
Finally, from Theorem~\ref{th:learning_MCCCE}, the following result follows:
\begin{corollary}\label{cor:easy_main}
	In instances $I := (\Gamma,\Phi)$ satisfying Assumption~\ref{ass:cce}, a constrained $\epsilon$-Phi-equilibrium can be computed in time polynomial in the instance size and $\frac{1}{\epsilon}$ by means of a decentralized learning algorithm.
\end{corollary}

\subsection{Marginally-constrained CCE}\label{sec:easy_cce}

We conclude the section by introducing a particular (natural) notion of constrained $\epsilon$-Phi-equilibrium for which Assumption~\ref{ass:cce} is satisfied.
This is a constrained version of the classical CCE in cost-constrained normal-form games where a player's costs \emph{only} depend on the action of that player.
We call it \emph{marginally-constrained} $\epsilon$-CCE.
Formally, such an equilibrium is defined for games in which, for every player $i \in \N$, it holds $\cvec_i(\avec) = \cvec_i(\avec')$ for all $\avec, \avec' \in \A$ such that $a_i = a_i'$, and for the set $\Phi_{\textsc{CCE}}$ of CCE deviations that we have previously introduced in Section~\ref{sec:prelim_relations}.
%
%
%Assumption~\ref{ass:cce} holds in this natural setting. Consider the deviation set defined by $\Phi_{i,,\textnormal{CCE}}$ introduced in Section~\ref{sec:prelim} and the fact that the costs $\Cvec_i$ dependes only on the marginalization of the correlated strategy, \ie it exists a function $\tilde\Cvec^\prime:\Delta(\A_i)\to[-1,1]^m$ such that $\Cvec_i(\zvec)=\tilde\Cvec^\prime(\xvec_i)$ where $\xvec_i$ is the marginalized strategy of the $i$-th player, \ie $\xvec_i(a_i)=\sum_{\avec_i\in\A^{-i}}\zvec[a_i,\avec_{-i}]$. We can prove that under these conditions Assumption~\ref{ass:cce} holds.
%
%
Next, we prove that, with the definition above, Assumption~\ref{ass:cce} is satisfied.
\begin{restatable}{theorem}{marginal}\label{thm:marginal}
	For instances $I := (\Gamma, \Phi_{\textnormal{CCE}})$ such that $\cvec_i(\avec) = \cvec_i(\avec')$ for every player $i \in \N$ and action profiles $\avec, \avec' \in \A: a_i = a_i'$, Assumption~\ref{ass:cce} holds.
	%	
	%	If $\Phi_i=\Phi_{i,\textnormal{CCE}}$ for all $i\in[N]$ and the costs $\Cvec_i(\zvec)$ only depend on the marginalization $\xvec_i$ of $\zvec$ for the $i$-th player, then Assumption~\ref{ass:cce} holds.
\end{restatable}

Thanks to Theorem~\ref{thm:marginal}, we readily obtain the two following corollaries of Theorems~\ref{th:optimal_MCCCE}~and~\ref{th:optimal_MCCCE}.
\begin{corollary}
	The problem of computing a marginally-constrained (exact) CCE that maximizes a linear function $\ell: \Delta_{\A} \to \mathbb{R}$ can be solved in polynomial time.
	%
	%	Restricted to instances $I := (\Gamma, \Phi_{\textnormal{CCE}})$ such that $\cvec_i(\avec) = \cvec_i(\avec')$ for every $i \in \N$ and $\avec, \avec' \in \A: a_i = a_i'$, the problem \textsc{ApxCPE}$(1,0)$ admits a polynomial-time algorithm.
	%	
	%	There exists a polynomial time algorithm to compute an exact optimal Marginalized Constrained Coarse Correlated Equilibrium, for any linear function $\ell:\Delta(\A)\to[0,1]$.
\end{corollary}
\begin{corollary}
	A marginally-constrained $\epsilon$-CCE can be computed in time polynomial in the instance size and $\frac{1}{\epsilon}$ by means of a decentralized learning algorithm.
	%
	%
	% There exists a decoupled learning dynamics that converges to a Marginalized Constrained Coarse Correlated Equilibrium.
\end{corollary}

	%\section{Discussion on a Related Open Problem}
%
\section{Discussion and Open Problems}

The main positive results that we provide in this paper (Corollaries~\ref{cor:res_1}~and~\ref{cor:res_2}) show that a constrained $\epsilon$-Phi equilibrium maximizing a given linear function can be computed in time polynomial in the instance size and $\frac{1}{\epsilon}$, when either the number of constraints or that of players' actions is fixed.
Clearly, this implies that, under the same assumptions, \emph{a} constrained $\epsilon$-Phi-equilibrium can be found efficiently.
Moreover, in Section~\ref{sec:easy}, we designed an efficient no-regret learning algorithm that finds \emph{a} constrained $\epsilon$-Phi-equilibrium in settings enjoying special properties (Corollary~\ref{cor:easy_main}).
%
%In our work we found a polynomial algorithm for finding approximate optimal Constrained Phi-equilibria in the case of constant number of actions or constant number of constraints. This clearly implies that we are able to find \emph{any} approximate Constrained Phi-equilibria in polynomial time under such assumptions.
%Furthermore, in Section~\ref{sec:easy} we designed an algorithm that finds an equilibria under the assumption that the costs depends only on the deviations.
%
However, the problem of efficiently computing \emph{a} constrained $\epsilon$-Phi-equilibrium remains open in general. Formally:
%
%
%Formally, we leave the following open problem:
%
\begin{definition}[Open Problem]\label{openProblem}
	Given any instance $I := (\Gamma,\Phi)$, find a constrained $\epsilon$-Phi-equilibrium in time polynomial in the instance size and $\frac{1}{\epsilon}$.
	%	
	%	Given a normal form game, sets $\Phi_i$ of deviations and constraints $\Cvec_i$ for $i\in[N]$ and $\epsilon>0$: find a correlated $\zvec\in\mcS$ in $\poly(1/\epsilon)$ iterations such that $U_i(\zvec)\ge\max_{\phi_i\in\Phi_i(\zvec)}U_i(\phi_i\diamond \zvec)-\epsilon$ for all $i\in[N]$.
\end{definition}
Solving the problem above is non-trivial.
Proposition~\ref{prop:non-convex} in Section~\ref{sec:hard} proves that the set of constrained $\epsilon$-Phi-equilibria is non-convex, and, thus, solving the problem in Definition~\ref{openProblem} is out of scope for most of the known equilibrium computation techniques.
On the other hand, it is unlikely that such a problem is \NPHard.
Indeed, a constrained $\epsilon$-Phi-equilibrium always exists and, given any $\zvec \in \Delta_{\A}$, it is possible to verify whether $\zvec$ is an equilibrium or not in polynomial time.
Formally, such a problem is said to belong to the $\mathsf{TFNP}$ complexity class, and, thus, standard arguments show that, if the problem is \NPHard, then $\mathsf{NP} = \mathsf{coNP}$~\citep{megiddo1991total}.
%
% See~\cite{bibid} for a similar argument.
%
Thus, one should try to reduce the problem in Definition~\ref{openProblem} to problems in $\mathsf{TFNP}$, such as that of computing a Nash equilibrium.
However, while the problem in Definition~\ref{openProblem} shares some properties with that of computing a Nash equilibrium, such as the non-convexity of the set of the equilibria, the former is fundamentally different from the latter, since it exhibits correlation among the players.
Thus, a reduction from such a problem to that of computing Nash equilibria would require a gadget to break the correlation among the players, and doing that is highly non-trivial as cost constraints are expressed by linear functions.

	\section*{Acknowledgements}
	We thank Denizalp Goktas that identified an inaccuracy in our previous definition of the set of safe deviations.
	
	\clearpage
	\bibliographystyle{ACM-Reference-Format}
	\bibliography{biblio}

%%% -*-BibTeX-*-
%%% Do NOT edit. File created by BibTeX with style
%%% ACM-Reference-Format-Journals [18-Jan-2012].

\begin{thebibliography}{27}

%%% ====================================================================
%%% NOTE TO THE USER: you can override these defaults by providing
%%% customized versions of any of these macros before the \bibliography
%%% command.  Each of them MUST provide its own final punctuation,
%%% except for \shownote{}, \showDOI{}, and \showURL{}.  The latter two
%%% do not use final punctuation, in order to avoid confusing it with
%%% the Web address.
%%%
%%% To suppress output of a particular field, define its macro to expand
%%% to an empty string, or better, \unskip, like this:
%%%
%%% \newcommand{\showDOI}[1]{\unskip}   % LaTeX syntax
%%%
%%% \def \showDOI #1{\unskip}           % plain TeX syntax
%%%
%%% ====================================================================

\ifx \showCODEN    \undefined \def \showCODEN     #1{\unskip}     \fi
\ifx \showDOI      \undefined \def \showDOI       #1{#1}\fi
\ifx \showISBNx    \undefined \def \showISBNx     #1{\unskip}     \fi
\ifx \showISBNxiii \undefined \def \showISBNxiii  #1{\unskip}     \fi
\ifx \showISSN     \undefined \def \showISSN      #1{\unskip}     \fi
\ifx \showLCCN     \undefined \def \showLCCN      #1{\unskip}     \fi
\ifx \shownote     \undefined \def \shownote      #1{#1}          \fi
\ifx \showarticletitle \undefined \def \showarticletitle #1{#1}   \fi
\ifx \showURL      \undefined \def \showURL       {\relax}        \fi
% The following commands are used for tagged output and should be
% invisible to TeX
\providecommand\bibfield[2]{#2}
\providecommand\bibinfo[2]{#2}
\providecommand\natexlab[1]{#1}
\providecommand\showeprint[2][]{arXiv:#2}

\bibitem[Altman and Shwartz(2000)]%
        {altman2000constrained}
\bibfield{author}{\bibinfo{person}{Eitan Altman} {and} \bibinfo{person}{Adam
  Shwartz}.} \bibinfo{year}{2000}\natexlab{}.
\newblock \showarticletitle{Constrained markov games: Nash equilibria}.
\newblock In \bibinfo{booktitle}{\emph{Advances in dynamic games and
  applications}}. \bibinfo{publisher}{Springer}, \bibinfo{pages}{213--221}.
\newblock


\bibitem[Alvarez-Mena and Hern{\'a}ndez-Lerma(2006)]%
        {alvarez2006existence}
\bibfield{author}{\bibinfo{person}{Jorge Alvarez-Mena} {and}
  \bibinfo{person}{On{\'e}simo Hern{\'a}ndez-Lerma}.}
  \bibinfo{year}{2006}\natexlab{}.
\newblock \showarticletitle{Existence of Nash equilibria for constrained
  stochastic games}.
\newblock \bibinfo{journal}{\emph{Mathematical Methods of Operations Research}}
  \bibinfo{volume}{63}, \bibinfo{number}{2} (\bibinfo{year}{2006}),
  \bibinfo{pages}{261--285}.
\newblock


\bibitem[Arrow and Debreu(1954)]%
        {arrow1954existence}
\bibfield{author}{\bibinfo{person}{Kenneth~J Arrow} {and}
  \bibinfo{person}{Gerard Debreu}.} \bibinfo{year}{1954}\natexlab{}.
\newblock \showarticletitle{Existence of an equilibrium for a competitive
  economy}.
\newblock \bibinfo{journal}{\emph{Econometrica: Journal of the Econometric
  Society}} (\bibinfo{year}{1954}), \bibinfo{pages}{265--290}.
\newblock


\bibitem[Aumann(1974)]%
        {aumann1974subjectivity}
\bibfield{author}{\bibinfo{person}{Robert~J Aumann}.}
  \bibinfo{year}{1974}\natexlab{}.
\newblock \showarticletitle{Subjectivity and correlation in randomized
  strategies}.
\newblock \bibinfo{journal}{\emph{Journal of mathematical Economics}}
  \bibinfo{volume}{1}, \bibinfo{number}{1} (\bibinfo{year}{1974}),
  \bibinfo{pages}{67--96}.
\newblock


\bibitem[Bakhtin et~al\mbox{.}(2022)]%
        {meta2022human}
\bibfield{author}{\bibinfo{person}{A Bakhtin}, \bibinfo{person}{N Brown},
  \bibinfo{person}{E Dinan}, \bibinfo{person}{G Farina}, \bibinfo{person}{C
  Flaherty}, \bibinfo{person}{D Fried}, \bibinfo{person}{A Goff},
  \bibinfo{person}{J Gray}, \bibinfo{person}{H Hu}, \bibinfo{person}{AP Jacob},
  {et~al\mbox{.}}} \bibinfo{year}{2022}\natexlab{}.
\newblock \showarticletitle{Human-level play in the game of Diplomacy by
  combining language models with strategic reasoning.}
\newblock \bibinfo{journal}{\emph{Science (New York, NY)}}
  (\bibinfo{year}{2022}), \bibinfo{pages}{eade9097--eade9097}.
\newblock


\bibitem[Bernasconi et~al\mbox{.}(2023)]%
        {bernasconi2023constrained}
\bibfield{author}{\bibinfo{person}{Martino Bernasconi}, \bibinfo{person}{Matteo
  Castiglioni}, \bibinfo{person}{Alberto Marchesi}, \bibinfo{person}{Francesco
  Trovo}, {and} \bibinfo{person}{Nicola Gatti}.}
  \bibinfo{year}{2023}\natexlab{}.
\newblock \showarticletitle{Constrained phi-equilibria}. In
  \bibinfo{booktitle}{\emph{International Conference on Machine Learning}}.
  PMLR, \bibinfo{pages}{2184--2205}.
\newblock


\bibitem[Brown and Sandholm(2019)]%
        {brown2019superhuman}
\bibfield{author}{\bibinfo{person}{Noam Brown} {and} \bibinfo{person}{Tuomas
  Sandholm}.} \bibinfo{year}{2019}\natexlab{}.
\newblock \showarticletitle{Superhuman AI for multiplayer poker}.
\newblock \bibinfo{journal}{\emph{Science}} \bibinfo{volume}{365},
  \bibinfo{number}{6456} (\bibinfo{year}{2019}), \bibinfo{pages}{885--890}.
\newblock


\bibitem[Bueno et~al\mbox{.}(2019)]%
        {bueno2019optimality}
\bibfield{author}{\bibinfo{person}{Luis~Felipe Bueno}, \bibinfo{person}{Gabriel
  Haeser}, {and} \bibinfo{person}{Frank~Navarro Rojas}.}
  \bibinfo{year}{2019}\natexlab{}.
\newblock \showarticletitle{Optimality conditions and constraint qualifications
  for generalized Nash equilibrium problems and their practical implications}.
\newblock \bibinfo{journal}{\emph{SIAM Journal on Optimization}}
  \bibinfo{volume}{29}, \bibinfo{number}{1} (\bibinfo{year}{2019}),
  \bibinfo{pages}{31--54}.
\newblock


\bibitem[Celli et~al\mbox{.}(2020)]%
        {celli2020no}
\bibfield{author}{\bibinfo{person}{Andrea Celli}, \bibinfo{person}{Alberto
  Marchesi}, \bibinfo{person}{Gabriele Farina}, {and} \bibinfo{person}{Nicola
  Gatti}.} \bibinfo{year}{2020}\natexlab{}.
\newblock \showarticletitle{No-regret learning dynamics for extensive-form
  correlated equilibrium}.
\newblock \bibinfo{journal}{\emph{Advances in Neural Information Processing
  Systems}}  \bibinfo{volume}{33} (\bibinfo{year}{2020}),
  \bibinfo{pages}{7722--7732}.
\newblock


\bibitem[Chen et~al\mbox{.}(2022)]%
        {chenfinding}
\bibfield{author}{\bibinfo{person}{Ziyi Chen}, \bibinfo{person}{Shaocong Ma},
  {and} \bibinfo{person}{Yi Zhou}.} \bibinfo{year}{2022}\natexlab{}.
\newblock \showarticletitle{Finding Correlated Equilibrium of Constrained
  Markov Game: A Primal-Dual Approach}. In \bibinfo{booktitle}{\emph{Advances
  in Neural Information Processing Systems}}.
\newblock


\bibitem[Daskalakis et~al\mbox{.}(2009)]%
        {daskalakis2009complexity}
\bibfield{author}{\bibinfo{person}{Constantinos Daskalakis},
  \bibinfo{person}{Paul~W Goldberg}, {and} \bibinfo{person}{Christos~H
  Papadimitriou}.} \bibinfo{year}{2009}\natexlab{}.
\newblock \showarticletitle{The complexity of computing a Nash equilibrium}.
\newblock \bibinfo{journal}{\emph{SIAM J. Comput.}} \bibinfo{volume}{39},
  \bibinfo{number}{1} (\bibinfo{year}{2009}), \bibinfo{pages}{195--259}.
\newblock


\bibitem[Ekeland and Temam(1999)]%
        {ekeland1999convex}
\bibfield{author}{\bibinfo{person}{Ivar Ekeland} {and} \bibinfo{person}{Roger
  Temam}.} \bibinfo{year}{1999}\natexlab{}.
\newblock \bibinfo{booktitle}{\emph{Convex analysis and variational problems}}.
\newblock \bibinfo{publisher}{SIAM}.
\newblock


\bibitem[Facchinei and Kanzow(2010)]%
        {facchinei2010generalized}
\bibfield{author}{\bibinfo{person}{Francisco Facchinei} {and}
  \bibinfo{person}{Christian Kanzow}.} \bibinfo{year}{2010}\natexlab{}.
\newblock \showarticletitle{Generalized Nash equilibrium problems}.
\newblock \bibinfo{journal}{\emph{Annals of Operations Research}}
  \bibinfo{volume}{175}, \bibinfo{number}{1} (\bibinfo{year}{2010}),
  \bibinfo{pages}{177--211}.
\newblock


\bibitem[Goktas and Greenwald(2022)]%
        {goktasexploitability}
\bibfield{author}{\bibinfo{person}{Denizalp Goktas} {and} \bibinfo{person}{Amy
  Greenwald}.} \bibinfo{year}{2022}\natexlab{}.
\newblock \showarticletitle{Exploitability Minimization in Games and Beyond}.
\newblock \bibinfo{journal}{\emph{Advances in Neural Information Processing
  Systems}} (\bibinfo{year}{2022}).
\newblock


\bibitem[Greenwald and Jafari(2003)]%
        {greenwald2003general}
\bibfield{author}{\bibinfo{person}{Amy Greenwald} {and} \bibinfo{person}{Amir
  Jafari}.} \bibinfo{year}{2003}\natexlab{}.
\newblock \showarticletitle{A general class of no-regret learning algorithms
  and game-theoretic equilibria}.
\newblock In \bibinfo{booktitle}{\emph{Learning theory and kernel machines}}.
  \bibinfo{publisher}{Springer}, \bibinfo{pages}{2--12}.
\newblock


\bibitem[Hakami and Dehghan(2015)]%
        {hakami2015learning}
\bibfield{author}{\bibinfo{person}{Vesal Hakami} {and} \bibinfo{person}{Mehdi
  Dehghan}.} \bibinfo{year}{2015}\natexlab{}.
\newblock \showarticletitle{Learning stationary correlated equilibria in
  constrained general-sum stochastic games}.
\newblock \bibinfo{journal}{\emph{IEEE Transactions on Cybernetics}}
  \bibinfo{volume}{46}, \bibinfo{number}{7} (\bibinfo{year}{2015}),
  \bibinfo{pages}{1640--1654}.
\newblock


\bibitem[H{\aa}stad(1999)]%
        {hastad1999clique}
\bibfield{author}{\bibinfo{person}{Johan H{\aa}stad}.}
  \bibinfo{year}{1999}\natexlab{}.
\newblock \showarticletitle{Clique is hard to approximate within
  $n^{1-\epsilon}$}.
\newblock \bibinfo{journal}{\emph{Acta Mathematica}} \bibinfo{volume}{182},
  \bibinfo{number}{1} (\bibinfo{year}{1999}), \bibinfo{pages}{105--142}.
\newblock
\showISBNx{1871-2509}


\bibitem[Hazan et~al\mbox{.}(2016)]%
        {hazan2016introduction}
\bibfield{author}{\bibinfo{person}{Elad Hazan} {et~al\mbox{.}}}
  \bibinfo{year}{2016}\natexlab{}.
\newblock \showarticletitle{Introduction to online convex optimization}.
\newblock \bibinfo{journal}{\emph{Foundations and Trends{\textregistered} in
  Optimization}} \bibinfo{volume}{2}, \bibinfo{number}{3-4}
  (\bibinfo{year}{2016}), \bibinfo{pages}{157--325}.
\newblock


\bibitem[Jordan et~al\mbox{.}(2022)]%
        {jordan2022first}
\bibfield{author}{\bibinfo{person}{Michael~I Jordan}, \bibinfo{person}{Tianyi
  Lin}, {and} \bibinfo{person}{Manolis Zampetakis}.}
  \bibinfo{year}{2022}\natexlab{}.
\newblock \showarticletitle{First-Order Algorithms for Nonlinear Generalized
  Nash Equilibrium Problems}.
\newblock \bibinfo{journal}{\emph{arXiv preprint arXiv:2204.03132}}
  (\bibinfo{year}{2022}).
\newblock


\bibitem[Kanzow and Steck(2016)]%
        {kanzow2016augmented}
\bibfield{author}{\bibinfo{person}{Christian Kanzow} {and}
  \bibinfo{person}{Daniel Steck}.} \bibinfo{year}{2016}\natexlab{}.
\newblock \showarticletitle{Augmented Lagrangian methods for the solution of
  generalized Nash equilibrium problems}.
\newblock \bibinfo{journal}{\emph{SIAM Journal on Optimization}}
  \bibinfo{volume}{26}, \bibinfo{number}{4} (\bibinfo{year}{2016}),
  \bibinfo{pages}{2034--2058}.
\newblock


\bibitem[Megiddo and Papadimitriou(1991)]%
        {megiddo1991total}
\bibfield{author}{\bibinfo{person}{Nimrod Megiddo} {and}
  \bibinfo{person}{Christos~H Papadimitriou}.} \bibinfo{year}{1991}\natexlab{}.
\newblock \showarticletitle{On total functions, existence theorems and
  computational complexity}.
\newblock \bibinfo{journal}{\emph{Theoretical Computer Science}}
  \bibinfo{volume}{81}, \bibinfo{number}{2} (\bibinfo{year}{1991}),
  \bibinfo{pages}{317--324}.
\newblock


\bibitem[Moulin and Vial(1978a)]%
        {moulin1978}
\bibfield{author}{\bibinfo{person}{H. Moulin} {and} \bibinfo{person}{J-P
  Vial}.} \bibinfo{year}{1978}\natexlab{a}.
\newblock \showarticletitle{Strategically zero-sum games: the class of games
  whose completely mixed equilibria cannot be improved upon}.
\newblock \bibinfo{journal}{\emph{INT J GAME THEORY}} \bibinfo{volume}{7},
  \bibinfo{number}{3} (\bibinfo{year}{1978}), \bibinfo{pages}{201--221}.
\newblock


\bibitem[Moulin and Vial(1978b)]%
        {moulin1978strategically}
\bibfield{author}{\bibinfo{person}{Herv{\'e} Moulin} {and} \bibinfo{person}{J-P
  Vial}.} \bibinfo{year}{1978}\natexlab{b}.
\newblock \showarticletitle{Strategically zero-sum games: the class of games
  whose completely mixed equilibria cannot be improved upon}.
\newblock \bibinfo{journal}{\emph{International Journal of Game Theory}}
  \bibinfo{volume}{7}, \bibinfo{number}{3} (\bibinfo{year}{1978}),
  \bibinfo{pages}{201--221}.
\newblock


\bibitem[Nash(1951)]%
        {nash1951non}
\bibfield{author}{\bibinfo{person}{John Nash}.}
  \bibinfo{year}{1951}\natexlab{}.
\newblock \showarticletitle{Non-cooperative games}.
\newblock \bibinfo{journal}{\emph{Annals of mathematics}}
  (\bibinfo{year}{1951}), \bibinfo{pages}{286--295}.
\newblock


\bibitem[Papadimitriou and Roughgarden(2008)]%
        {papadimitriou2008computing}
\bibfield{author}{\bibinfo{person}{Christos~H Papadimitriou} {and}
  \bibinfo{person}{Tim Roughgarden}.} \bibinfo{year}{2008}\natexlab{}.
\newblock \showarticletitle{Computing correlated equilibria in multi-player
  games}.
\newblock \bibinfo{journal}{\emph{Journal of the ACM (JACM)}}
  \bibinfo{volume}{55}, \bibinfo{number}{3} (\bibinfo{year}{2008}),
  \bibinfo{pages}{1--29}.
\newblock


\bibitem[Rosen(1965)]%
        {rosen1965existence}
\bibfield{author}{\bibinfo{person}{J~Ben Rosen}.}
  \bibinfo{year}{1965}\natexlab{}.
\newblock \showarticletitle{Existence and uniqueness of equilibrium points for
  concave n-person games}.
\newblock \bibinfo{journal}{\emph{Econometrica: Journal of the Econometric
  Society}} (\bibinfo{year}{1965}), \bibinfo{pages}{520--534}.
\newblock


\bibitem[Zuckerman(2007)]%
        {Zuckerman2007linear}
\bibfield{author}{\bibinfo{person}{David Zuckerman}.}
  \bibinfo{year}{2007}\natexlab{}.
\newblock \showarticletitle{Linear Degree Extractors and the Inapproximability
  of Max Clique and Chromatic Number}.
\newblock \bibinfo{journal}{\emph{Theory of Computing}} \bibinfo{volume}{3},
  \bibinfo{number}{6} (\bibinfo{year}{2007}), \bibinfo{pages}{103--128}.
\newblock


\end{thebibliography}
	%\bibliographystyle{icml2023}
	
	% applications?\part{title}
	
	\clearpage
	\onecolumn
	\appendix
	
\section{On the Weaknesses of the Guarantees of the Algorithm of~\citet{chenfinding}}\label{app:weakness}
The Algorithm of~\citet{chenfinding} finds a distribution $\mu$ over correlated strategies  $\Delta_{\A}$ such that:
\begin{equation}\label{eq:expectation_IC}
\mathbb{E}_{\zvec\sim\mu}\left[\max\limits_{\phi_i\in\Phi_i^\mcS(\zvec)}u_i(\phi_i\diamond\zvec)-u_i(\zvec)\right]\le 0.
\end{equation}
However, here we claim that this solution concept inherits some weaknesses from the non-convexity of the equilibria set that we proved in Theorem~\ref{thm:convex}.
Indeed, consider the same instance of Theorem~\ref{thm:convex} and consider the uniform distribution $\mu$ over $\{\zvec_1, \zvec_2\}$. In Theorem~\ref{thm:convex} we proved that $\max_{\phi_i\in\Phi_i^\mcS(\zvec^1)}u_i(\phi_i\diamond\zvec^1)-u_i(\zvec^1)\le 0$ for all $i\in\{1,2\}$ and $\max_{\phi_i\in\Phi_i^\mcS(\zvec^2)}u_i(\phi_i\diamond\zvec^2)-u_i(\zvec^2)\le 0$ for all $i\in\{1,2\}$ and thus Equation~\eqref{eq:expectation_IC} holds over the distribution $\mu$.

However we show that the expected correlated strategy $\zvec^3$ derived from distribution $\mu$, \ie $\zvec^3=\mathbb{E}_{\zvec\sim\mu}[\zvec]=\frac{1}{2}\zvec^1+\frac{1}{2}\zvec^2$, it is not a feasible equilibrium, or an approximate one.

Indeed, in Theorem~\ref{thm:convex}, we proved that $\max_{\phi_2\in\Phi_2^\mcS(\zvec^3)}u_2(\phi_2\diamond\zvec^3)-u_2(\zvec^3)\ge\frac{1}{3}$, showing that the average correlated strategies returned by their Algorithm is not an equilibrium nor close to it.

This comes from the peculiar fact about Constrained Phi-equilibria that exhibit non-convex set of solutions, which is in striking contrast with the unconstrained case.
Indeed the guarantees of Equilibria~\eqref{eq:expectation_IC} would imply that $\mathbb{E}_{\zvec\sim\mu}[\zvec]$ is a equilibrium in the unconstrained case in which the set of equilibria is convex.

\section{Proofs Omitted from Section~\ref{sec:prelim} }

\existence*
\begin{proof}
	With assumption~\ref{ass:strictly} \citet[Theorem~2.1]{altman2000constrained} proves the existence of a constrained Nash equilibrium. In our setting this is equivalent to a product distribution $\zvec=\otimes_{i\in[N]}\xvec_i$ so that it is a Constrained Phi-equilibrium for any set of deviations $\Phi_i$.\footnote{As common in the normal form game literature, for any distribution $\xvec\in\Delta(X)$ and $\yvec\in\Delta(Y)$, $\xvec\otimes\yvec\in\Delta(X\times Y)$ is the product distribution defined as $(\xvec\otimes\yvec)[a,b]=\xvec[a]\yvec[b]$ for $a\in X$ and $b\in Y$.}
	This is easily seen by observing that a Constrained Nash Equilibria is defined as:
	\[
	\sum\limits_{\avec\in\A}u_i\left(\prod_{j\in[N]}\xvec_j(a_j)\right)\ge 	\sum\limits_{\avec\in\A}u_i\left(\tilde\xvec_i(a_j)\prod_{j\in[N]\setminus \{i\}}\xvec_i(a_i)\right)
	\]
	for all $\tilde\xvec_i\in\Delta(\A^i)$ s.t. $\xvec_i\otimes\xvec_{-i}\in\mcS$.
	
	On the other hand it easily seen that for all $\phi_i\in\Phi_i(\zvec)$ there exists some $\tilde\xvec_i\in\Delta(\A^i)$ such that 
	\[
	\phi_i\diamond\left( \otimes_{j\in[N]}\xvec_{j}\right) = \tilde\xvec_i\otimes\xvec_{-i}
	\]
	and $\tilde\xvec_i\otimes\xvec_{-i}\in\mcS$.
	
	This is proved by the following calculations:
	\begin{align}
		\phi_i\diamond\left( \otimes_{j\in[N]}\xvec_{j}\right)[a_i,\avec_{-i}]&:=\sum\limits_{b\in\A^i}\phi_i[b,a_i]\xvec_i(b)\xvec_{-i}(\avec_{-i})\\
		&=\tilde\xvec_i(a_i)\otimes\xvec_{-i}(\avec_{-i}),
	\end{align}
	where $\tilde\xvec_i(a_i):=\sum_{b\in\A^i}\phi_i[b,a_i]\xvec_i(b)$ and $\tilde\xvec_i\in\Delta(A^i)$ since, by definition, $\sum_{a_i\in\A^i}\phi_i[b,a_i]=1$ for all $b\in\A^i$.
	
	This proves that a Constrained Nash Equilibrium is a Phi-Constrained Equilibrium for all $\Phi$. 
\end{proof}
	%\section{Proofs Omitted from Section~\ref{sec:hard} }

%\subsection{Hardness Result of~\textsc{ApxCPE}}

\section{Proofs Omitted from Section~\ref{sec:hard}}

\theoremHardness*

\begin{proof}
We reduce from \textsc{GAP-INDEPENDENT-SET}, which is a promise problem that formally reads as follows: given an $\delta > 0$ and a graph $G= (V,E)$, with set of nodes $V$ and set of edges $E$, determine whether $G$ admits an independent set of size at least $|V|^{1-\delta}$ or all the independent sets of $G$ have size smaller than $|V|^{\delta}$. 
\textsc{GAP-INDEPENDENT-SET} is \NPHard\ for every $\delta>0$~\citep{hastad1999clique,Zuckerman2007linear}.

 Let $\ell=|V|$ and $\alpha>0$ be the desired approximation factor. Given an instance of \textsc{GAP-INDEPENDENT-SET}, we build an instance such that if there exists an independent set of size $\ell^{1-\delta}$, then there exists a Constrained Phi-equilibrium with social welfare 1.
Otherwise, if all the independent sets have size at most $\ell^\delta$, all the Constrained $\epsilon$-Phi-equilibria have social welfare at most $\alpha/2$.
We can use any $\delta>0$, since we simply need $\ell^\delta<\ell^{1-\delta}$.
Moreover, we take $\epsilon=\frac{\alpha^2}{128 \ell^2}$. As we will see,  $\ell$ will be smaller than the number of action of the players, satisfying the condition in the statement.
	\paragraph{Construction.} 
	The first player has a set of actions $\A_1$ that includes actions $a_0$, $a_1$, $a_2$ and an action $a_v$ for each $v \in V$. Moreover, the first player has an action $a_F$.\footnote{This action is needed only to satisfy the strictly feasibility assumption.} 
	The second player has a set of actions $\A_2$ that includes actions $a_v$ and $\bar a_{v}$ for each $v \in V$. Moreover, the second player has an action $a_F$. 
	Let $\gamma=\eta=\alpha/8$.
	The utility of the first agent is as follows:
	\begin{itemize}
	\item $u_1(a_0,a)=\gamma+\frac{1}{2}\eta$ for all $a \in \A_2 \setminus \{a_F\}$,
	\item $u_1(a_1,a_v)=\gamma+\eta$ and $u_1(a_2,a_v)=\gamma$ for all $v \in V$.
	 \item $u_1(a_1,\bar a_v)=\gamma$ and $u_1(a_2,\bar a_v)=\gamma+\eta$ for all $v \in V$.
	 \item $u_1(a_v, a_{v})=u_1(a_v,\bar a_{v})=\gamma$ for all $v \in V$
	 \item $u_1(a_{v}, a_{v'})=\gamma$ and  $u_1(a_{v},\bar a_{v'})=\gamma+\frac{\ell-\ell^{1-\delta}}{\ell-\ell^{1-\delta}-1}\eta$ for all $v'\neq v$.
	 \item $u_1(a_F,a)=0$ for each $a \in \A_2$.
	 	 \item $u_1(a,a_F)=0$ for each $a \in \A_1$.
 	 \end{itemize}
 	The utility of the second agent is $u^2(a_0,a)=1$ for each $a \in \A_2 \setminus \{a_F\}$ and $0$ otherwise.
	 
	 There is a cost function $c_v$ for each $v \in V$, which is common to both the agents.
	 For each $v \in V$, the cost function $c_v$ is such that
	 \begin{itemize}
	 	\item $c_v(a_v, a_{v'})=-1$ for each $v'\neq v$, $(v,v') \in E$, \item $c_v(a_v, a_{v'})=0$ for each $v' \neq v, (v,v') \notin E$, 
	 	\item $c_v(a_v, a_{v})=1$ for each $v \in V$.
	 	\item $c_v(a_F,a)=-\frac{1}{4\ell^2}$ for each $a \in \A_2$. 
	 	\item $c_v(a,a_F)=-\frac{1}{4\ell^2}$ for each $a \in \A_1$.
	 	\item For every other action profile the cost is $0$.
	 \end{itemize}
	We dropped the player index from the cost functions $c$ as they are equal to both players.
 	
 	Moreover, we set of deviations $\Phi_i=\Phi_{i,\textsc{all}}$ for both players $i\in\{1,2\}$.
 	
 	Notice that the instance satisfies Assumption~\ref{ass:strictly}. Indeed, the deviation $\phi_i$ such that $\phi_i[a,a_F]=1$ for all $a\in\A_{i}$ for $i\in\{1,2\}$, that deviates deterministically to $a_F$ is always strictly feasible for both player $1$ and player $2$. Moreover, its cost is polynomial in the instance size.

	\paragraph{Completeness.}
	We show that if there exists an independent set of size $\ell^{1-\delta}$, then the social welfare of an optimal Constrained Phi-equilibria is at least $1$.
	Let $V^*$ be an independent set of size $\ell^{1-\delta}$. We build a Constrained Phi-equilibria $\zvec$ with social welfare at least $1$. Consider the correlated strategy such that $\zvec[a_0,a_v]=\frac{1}{2\ell^{1-\delta}}$ for all $v \in V^*$, while $\zvec[a_0,\bar a_v]=\frac{1}{2(\ell-\ell^{1-\delta})}$ for all $v \notin V^*$. All the other action profiles have probability $0$.
	
	It is easy to see that the correlated strategy has social welfare at least $1$ since player $1$ always plays action $a_0$ and $u^2(a_0,a)=1$ for all $a \in \A_2$. Moreover, it is easy to verify that it is safe since $c_v(a_0,a)\le 0$ for each $a \in \A_2$.
	Hence, to show that $z$ is an Constrained Phi-equilibria we only need to prove that it satisfies the incentive constraints.
	The incentive constraints of the second player are satisfied since they obtain the maximum possible utility, \emph{i.e.}, $1$. 
	
	Consider now a possible deviation of the first player $\phi_1\in\Phi_1$.
	As a first step, we compute the expected utility of a strategy $\phi_1$.
	Let us define the following quantities:
	\begin{itemize}
		\item $T^1=\sum_{v \in V^*}  \phi_1[a_0,a_v] \left[ \left(\zvec[a_0,a_v] +\zvec[a_0,\bar a_v] +\sum_{v' \neq v} \zvec[a_0,a_{v'}] \right) \gamma + \left(\gamma +\frac{\ell-\ell^{1-\delta}}{\ell-\ell^{1-\delta}-1}\eta \right) \sum_{v' \neq v}  \zvec[a_0,\bar a_v]          \right] $ 
		\item $T^2=\sum_{v \notin V^*} \phi_1[a_0,a_v] \left[ \left(\zvec[a_0,a_v] +\zvec[a_0,\bar a_v] +\sum_{v' \neq v} \zvec[a_0,a_{v'}] \right) \gamma+ \left(\gamma +\frac{\ell-\ell^{1-\delta}}{\ell-\ell^{1-\delta}-1}\eta\right)  \sum_{v' \neq v}  \zvec[a_0,\bar a_v]          \right] $
		\item  $T^3 = \left(\gamma + \frac{\eta}{2}\right)  \phi_1[a_0,a_0] +  \frac{\gamma + \eta}{2} (\phi_1[a_0,a_1]+\phi_1[a_0,a_2])+ \frac{\gamma}{2} (\phi_1[a_0,a_1]+\phi_1[a_0,a_2]) $
	\end{itemize}

	We bound each component individually.
	\begin{align*}
	T^1&=\sum_{v \in V^*}  \phi_1[a_0,a_v] \left[ \left(\zvec[a_0,a_v] +\zvec[a_0,\bar a_v] +\sum_{v' \neq v} \zvec[a_0,a_{v'}] \right)\gamma+  \left(\gamma + \eta\frac{\ell-\ell^{1-\delta}}{\ell-\ell^{1-\delta}-1}\right) \sum_{v' \neq v}  \zvec[a_0,\bar a_v]          \right] \\
	&=\sum_{v \in V^*}  \phi_1[a_0,a_v] \left[\frac{1}{2}\gamma + \frac{1}{2}\left(\gamma + \eta\frac{\ell-\ell^{1-\delta}}{\ell-\ell^{1-\delta}-1}\right) \right] \\
	&=\sum_{v \in V^*}  \phi_1[a_0,a_v] \left(\gamma+\frac{\eta}{2}\frac{\ell-\ell^{1-\delta}}{\ell-\ell^{1-\delta}-1}\right) \\
	&\le \sum_{v \in V^*}  \phi_1[a_0,a_v] (\gamma+\eta) ,
	\end{align*}
	where in the last inequality we use  $\frac{\ell-\ell^{1-\delta}}{\ell-\ell^{1-\delta}-1}\le 2$ for $\ell$ large enough.
	while
	\begin{align*}
	T^2&=\sum_{v \notin V^*}  \phi_1[a_0,a_v] \left[ \left(\zvec[a_0,a_v] +\zvec[a_0,\bar a_v] +\sum_{v' \neq v} \zvec[a_0,a_{v'}] \right) \gamma + \left(\gamma +\eta\frac{\ell-\ell^{1-\delta}}{\ell-\ell^{1-\delta}-1}  \right) \sum_{v' \neq v}  \zvec[a_0,\bar a_v]          \right] \\
	&= \sum_{v \notin V^*}  \phi_1[a_0,a_v] \left[\left(\frac{1}{2}+\frac{1}{2(\ell-\ell^{1-\delta})}\right)\gamma + \left(\frac{1}{2}-\frac{1}{2(\ell-\ell^{1-\delta})}\right) \left(\gamma +\eta\frac{\ell-\ell^{1-\delta}}{\ell-\ell^{1-\delta}-1} \right)\right]\\
	&=\sum_{v \notin V^*}  \phi_1[a_0,a_v]\left[\gamma +\frac{\eta}{2}\left(\frac{\ell-\ell^{1-\delta}}{\ell-\ell^{1-\delta}-1}-\frac{1}{\ell-\ell^{1-\delta}-1} \right)\right]\\
	&=\sum_{v \notin V^*}  \phi_1[a_0,a_v]\left(\gamma + \frac{\eta}{2}\right).
	\end{align*}
	Finally, 
	\begin{align*}
	T^3& = [a_0,a_0] \left(\gamma + \frac{\eta}{2}\right) + \frac{\gamma + \eta}{2}([a_0,a_1]+[a_0,a_2]) +\frac{\gamma}{2}([a_0,a_1]+[a_0,a_2])  \\
	&= \left( \gamma + \frac{\eta}{2}\right)\left([a_0,a_0] + \phi_1[a_0,a_1]+\phi_1[a_0,a_2]\right)
	\end{align*}
	
	Finally, the utility of a deviation $\phi_1$ is
	\begin{align*}
	&\sum_{a^1 \in \A_1,a^2\in \A_2} \sum_{a \in A_1} \phi_1[a^1,a] \zvec[a^1,a^2] u_1(a,a^2)\\
	&=  \sum_{a \in \A_1,a^2 \in \A_2} \phi_1[a_0,a] \zvec[a_0,a^2]u_1(a,a^2)\\
	&=T^1+T^2+T^3  \\
	&\le (\gamma+\eta) \sum_{v \in V^*}  \phi_1[a_0,a_v]  + \left(\gamma + \frac{\eta}{2} \right) \sum_{v \notin V^*}  \phi_1[a_0,a_v] +\left( \gamma + \frac{\eta}{2}\right)(\phi[a_0,a_0] + \phi_1[a_0,a_1]+\phi_1[a_0,a_2]) \\
	&=\frac{\eta}{2}\sum_{v \in V^*}  \phi_1[a_0,a_v]  +  \left(\gamma +\frac{\eta}{2}\right)(1- \phi_1[a_0,a_F] )
	\end{align*}
	
	Now, we show that no deviation $\phi_1\in\Phi_1$ is both safe and increases player $1$ utility. In particular, we show that if a strategy $\phi_1$ increases the utility than it is not safe.
	Indeed, if $\phi_1$ increases the utility, then
	\begin{align*}
	&\sum_{\substack{a^1 \in \A_1,\\a^2\in \A_2}} \sum_{a \in A_1} \phi_1[a^1,a] \zvec[a^1,a^2] u_1(a,a^2)> \gamma+ \frac{\eta}{2}
	\end{align*}
	This implies that 
		\begin{align*}
	&\frac{\eta}{2}\sum_{v \in V^*}  \phi_1[a_0,a_v]  +   \left(\gamma +\frac{\eta}{2}\right)(1- \phi_1[a_0,a_F] )  > \gamma+ \frac{\eta}{2}
	\end{align*}
	and 
	\begin{align}\label{eq:vLarge}
	&\sum_{v \in V^*}  \phi_1[a_0,a_v]    > \frac{1}{2}  \phi_1[a_0,a_F] 
	\end{align}

	Next, we show that any $\phi_1$ that increases the utility (and hence that satisfies Eq~\eqref{eq:vLarge}) is not a feasible deviation.
	First, notice that equation~\eqref{eq:vLarge} implies that there is a $\bar v\in V^*$ such that 
	\begin{align} \label{eq:vLarge2}
	&\phi_1[a_0,a_{\bar v}] > \frac{1}{2\ell}   \phi_1[a_0,a_F].
	\end{align}
	Then, we show that the deviation $\phi_1$ violates the constraint $c_{\bar v}$.
	 In particular,
	\begin{align*}
	\sum_{a^1 \in \A_1,a^2\in \A_2} \sum_{a \in \A_1} \phi_1[a^1,a] \zvec[a^1,a^2] c_v(a,a^2)&=  \phi_1[a_0,a_{\bar v}] \zvec[a_0,a_{\bar v}] 1 - \frac{1}{4\ell^2}\phi_1[a_0,a_F] - \sum_{\substack{v \in V^*: (v,\bar v) \in E }} \phi_1[a_0,a_{\bar v}] \zvec[a_0,a_v] 1  \\
	&=  \phi_1[a_0,a_{\bar v}] \zvec[a_0,a_{\bar v}]  - \frac{1}{4\ell^2} \phi_1[a_0,a_F]  \\
	 &=  \frac{1}{2\ell^{1-\frac{1}{\ell}}} \phi_1[a_0,a_{\bar v}]  -\frac{1}{4\ell^2}  \phi_1[a_0,a_F] \\
	&> \phi_1[a_0,a_{\bar v}] \left(\frac{1}{2\ell^{1-\frac{1}{\ell}}} -  \frac{1}{2\ell} \right)   \ge 0,
	\end{align*}
	where the second inequality holds since $V^*$ is an independent set, and the second-to-last inequality by Equation~\eqref{eq:vLarge2}.
	Hence, there is no deviation $\phi_1$ that increases players $1$ utility and that is safe. This concludes the first part of the proof.

	\textbf{Soundness.} We show that if there exists a Constrained $w$-Phi-equilibria with social welfare $\alpha/2$, then there exists an independent set of size strictly larger than  $\ell^{\delta}$, reaching a contradiction.  
	Suppose by contradiction that there exists a Constrained $\epsilon$-Phi-equilibrium $\zvec$ with social welfare strictly greater than $\alpha/2$. 
	Thus, 
	\[
	\sum_{a' \in \A_2 \setminus \{a_F\}} \zvec[a_0,a'] \cdot1 +\sum_{a \in \A_1,a' \in \A_2} (\gamma+\eta) \ge  \sum_{a \in \A_1,a' \in \A_2} \zvec[a,a'] (u_1(a,a')+u^2(a,a')) \ge \alpha/2,
	\] 
	where the first inequality comes from $u_2(a_0,a')=1$ for each $a'\in \A_2 \setminus \{a_F\}$ and $0$ otherwise, and $u_1(a,a')\le \gamma+\eta$ for each $a\in \A_1$ and $a' \in \A_2$.
	This implies 
	\begin{align}\label{eq:highProb}
	\sum_{a' \in \A_2} \zvec[a_0,a'] \ge \alpha/4.
	\end{align}

  	Then, we show that $\zvec$ assigns similar probabilities on the set of action profiles $\{a_0,a_v\}_{v \in V}$ and $\{a_0,\bar a_v\}_{v \in V}$
	Given an $a \in \A_1$, let $\phi_{a}\in\Phi_1$ be a deviation of the first player such that $\phi_{a}[a_0,a]=1$ and $\phi_{a}[a',a']=1$ for each $a' \neq a_0$.
	Since $\zvec$ is an Constrained $\epsilon$-Phi-equilibrium there is no feasible deviation $\phi_{a}$ that increases the utility of player $1$ by more than $\epsilon$.
	This implies that 
	\begin{align}\label{eq:balance2}
		\left|\sum_{v \in V} \zvec[a_0,a_v]-\sum_{v \in V} \zvec[a_0,\bar a_v]\right|\le \frac{2\epsilon}{\eta}.
	\end{align}
	Indeed, if 
	\begin{align}\label{eq:balance}
		&\sum_{v \in V} \zvec[a_0,a_v]> \sum_{v \in V} \zvec[a_0,\bar a_v]+ \frac{2\epsilon}{\eta},
	\end{align}
	 then the deviation  $\phi_{a_1}$ has utility at least
	\begin{align*}
		\sum_{v \in V} \zvec[a_0,a_v] &\phi_{a_1}[a_0,a_1] (\gamma+\eta) + \zvec[a_0,\bar a_v] \phi_{a_1}[a_0,a_1] \gamma +   \sum_{a \in \A_1 \setminus \{a_0\}} \sum_{ a' \in \A_2} \zvec[a, a'] \phi_{a_1}[a,a] u^i(a, a')\\
		&=  \eta \sum_{v \in V} \zvec[a_0,a_v]  +  \gamma \sum_{v \in V} \left(\zvec[a_0, a_v] + \zvec[a_0,\bar a_v] \right) +   \sum_{a \in \A_1 \setminus \{a_0\}} \sum_{ a' \in \A_2} \zvec[a, a'] \phi_{a_1}[a,a] u^i(a, a')\\
		&> \frac{\eta}{2} \left(\frac{2\epsilon}{\eta}+\sum_{v \in V} (\zvec[a_0,a_v] + \zvec[a_0,\bar a_v])\right)+  \gamma \sum_{v \in V} \left(\zvec[a_0, a_v] + \zvec[a_0,\bar a_v] \right)  \\
		&\hspace{2cm}+ \sum_{a \in \A_1 \setminus \{a_0\}} \sum_{ a' \in \A_2} \zvec[a, a'] \phi_{a_1}[a,a] u^i(a, a')\\
		& \ge \epsilon+ \left(\frac{\eta}{2}+\gamma\right)\sum_{v \in V} (\zvec[a_0,a_v] + \zvec[a_0,\bar a_v])+   \sum_{a \in \A_1 \setminus \{a_0\}} \sum_{ a' \in \A_2} \zvec[a, a'] \phi_{a_1}[a,a] u^i(a, a')\\
		& \ge u_1(z)+ \epsilon,
	\end{align*}
	where the first inequality comes from adding $\sum_{v \in V} \zvec[a_0,a_v]$ to both sides of Equation~\eqref{eq:balance}.
	Moreover, $\phi_{a_1}$ is feasible since for each constraint $c_{\bar v}$, $\bar  v \in V$, it has cost
	\begin{align*}
	&\sum_{v \in V} \left(\zvec[a_0,a_v] \phi_{a_1}[a_0,a_1] c_{\bar v}(a_1,a_v) + \zvec[a_0,\bar a_v] \phi_{a_1}[a_0,a_1] c_{\bar v}(a_1,a_v)\right) \\  
	&\hspace{2cm}+\sum_{a \in \A_1 \setminus \{a_0\}} \sum_{ a' \in \A_2} \zvec[a, a'] \phi_{a_1}[a,a] c_{\bar v}(a, a')\\
	&= \sum_{v \in V} \left(\zvec[a_0,a_v] \phi_{a_1}[a_0,a_1] c_{\bar v}(a_0,a_v) + \zvec[a_0,\bar a_v] \phi_{a_1}[a_0,a_1] c_{\bar v}(a_0,\bar a_v)\right) \\  
	&\hspace{2cm}+\sum_{a \in \A_1 \setminus \{a_0\}} \sum_{ a' \in \A_2}\zvec[a, a'] \phi_{a_1}[a,a] c_{\bar v}(a, a')\\
	&= c_{\bar v}(z)\le 0.
	 \end{align*}
	A similar argument shows that if $\sum_{v \in V} \zvec[a_0,a_v] < \sum_{v \in V} \zvec[a_0,\bar a_v]-\frac{2\epsilon}{\eta}$ then the deviation  $\phi_{a_2}$ is safe and increases the utility. 
	As a consequence of Equation~\eqref{eq:balance2}, it holds
	\begin{align}\label{eq:balance3}
		2\sum_{v \in V} \zvec[a_0,\bar a_v] \ge\sum_{v \in V} (\zvec[a_0, a_v]+\zvec[a_0,\bar a_v] ) -\frac{\epsilon}{\eta}= \sum_{a \in \A_2\setminus\{a_F\}} \zvec[a_0, a] -\frac{2\epsilon}{\eta},
	\end{align}
	where the first inequality comes from adding $\sum_{v \in V} \zvec[a_0,\bar a_v]$ to both sides of $\sum_{v \in V} \zvec[a_0,\bar a_v]|\ge \sum_{v \in V} \zvec[a_0,a_v]- \frac{2\epsilon}{\eta}$
	
	The next step is to show that it is if there is no safe deviation $\phi_{a_v}$, $v \in V$, that increases the utility, then there exists an independent set of size larger than $\ell^{\delta}$. 
	 Since $z$ is an Constrained $\epsilon$-Phi-equilibrium, for each $a_v$, $v\in V$ one of the following two conditions holds: i) $\phi_{a_v} \notin \Phi^\mcS_1(z)$ or ii) $u_1(\phi_{a_v}\diamond z)\le u_1(z) + \epsilon$.
	Let $V^1\subseteq V$ be the set of vertexes $v$ such that $\phi_{a_v}$ is not safe, \emph{i.e.}, $\phi_{a_v} \notin \Phi^\mcS_1(z)$,  and $V^2=V\setminus V^1$ be the set of $v$ such that $\phi_{a_v}$ does not increase the utility by more  than $\epsilon$ and are not in $V^1$, \emph{i.e.}, $u_1(\phi_{a_v}\diamond z)\le u_1(z)$ and $\phi_{a_v} \in \Phi_1^\mcS(z)$.
	We show that $|V^2|\le \ell-\ell^{1-\delta}$.
	Indeed, for each $v \in V^2$, deviation $\phi_{a_v}$ does not increase the utility and hence it holds:
	\begin{align*}
		\gamma \sum_{a \in \A_2 \setminus \{a_F\}} \zvec_[a_0,a]&  + \eta\frac{\ell-\ell^{1-\delta}}{\ell-\ell^{1-\delta}-1}   \sum_{v' \neq v}  \zvec[a_0,\bar a_{v'}]  +  \sum_{a \in A_1 \setminus \{a_0\}} \sum_{ a' \in \A_2} \zvec[a, a'] \phi_{a_1}[a,a] u^i(a, a')    \\
		&=\left(\sum_{v' \in V } \zvec[a_0,a] + \zvec[a_0,\bar a_v]\right) \phi[a_0,a_v] \gamma +\sum_{v' \neq v} \phi[a_0,a_v] \zvec[a_0,\bar a_{v'}]  \left(\gamma+\eta\frac{\ell-\ell^{1-\delta}}{\ell-\ell^{1-\delta}-1}  \right) \\ 
		&\hspace{2cm}+\sum_{a \in \A_1 \setminus \{a_0\}} \sum_{ a' \in \A_2} \zvec[a, a'] \phi_{a_1}[a,a] u^i(a, a') \\ &\le u_1(z) +\epsilon\\
		&  =\left(\gamma+ \frac{\eta}{2}\right) \sum_{a \in \A_2 \setminus \{a_F\}} \zvec[a_0,a] +  \sum_{a \in \A_1 \setminus \{a_0\}} \sum_{ a' \in \A_2} \zvec[a, a']  u^i(a, a')   +\epsilon,\\
	\end{align*}
	where the inequality holds since the lhs is the utility of the deviation $\phi_{a_v}$.
	
	This implies
	\[ 
	\left(\sum_{v' }  \zvec[a_0,\bar a_{v'}]-  \zvec[a_0,\bar a_{v}] \right)\eta\frac{\ell-\ell^{1-\delta}}{\ell-\ell^{1-\delta}-1} \le \frac{\eta}{2}  \sum_{a \in \A_2 \setminus \{a_F\}}  \zvec[a_0, a] +\epsilon \le \eta  \sum_{v \in V}  \zvec[a_0, \bar a_v] + 2\epsilon  ,
	\] 
	where the last inequality holds by Equation~\eqref{eq:balance3}.
	Hence, 
	\[     
	\zvec[a_0,\bar a_v] \frac{\ell-\ell^{1-\delta}}{\ell-\ell^{1-\delta}-1} \ge \left(\frac{\ell-\ell^{1-\delta}}{\ell-\ell^{1-\delta}-1}-1\right)\sum_{v' }  \zvec[a_0,\bar a_{v'}]  -2\epsilon/\eta,    
	\]
	and
	\begin{align}\label{eq:large}
		   \bar \zvec[a_0,a_v]  \ge \frac{1}{\ell-\ell^{1-\delta}} \sum_{v' }  \zvec[a_0,\bar a_{v'}] -2\epsilon/\eta .  
	\end{align}
	
	Suppose that $|V^2|>\ell-\ell^{1-\delta}$, and hence Equation~\eqref{eq:large} is satisfied by at least $|V^2|\ge \ell-\ell^{1-\delta}+1 $ vertexes.
	We need the following inequality. 
	\begin{align} \label{eq:contr}
		\frac{1}{\ell} \sum_{v' }  \zvec[a_0,\bar a_{v'}] &\ge \frac{1}{\ell} \sum_{a \in \A_2\setminus\{a_F\}} \zvec[a_0, a] -\frac{2\epsilon}{\ell\eta}
		\ge \frac{\alpha}{4\ell}-\frac{2\epsilon}{\ell\eta}
		= \frac{\alpha}{4\ell}-\frac{\alpha}{8\ell^3} 
		\ge  \frac{\alpha}{8\ell} 
		=\frac{2\ell}{\eta} \left(\frac{\alpha^2}{16\ell^2 }  \right)
			=\frac{2\ell}{\eta} \epsilon 
	\end{align}
	where the first inequality comes from Equation~\eqref{eq:balance3}, and the second one by Equation~\eqref{eq:highProb}.
	Then, summing over the $|V^2|$ inequalities we get
	\begin{align*}
		\sum_{v \in V^2} \bar \zvec[a_0,a_v]&\ge  (\ell-\ell^{1-\delta}+1 ) \left(\frac{1}{\ell-\ell^{1-\delta}} \sum_{v' }  \zvec[a_0,\bar a_{v'}] -2\epsilon/\eta\right) \\
		& \ge \sum_{v' }  \zvec[a_0,\bar a_{v'}] + \frac{1}{\ell} \sum_{v' } \zvec[a_0,\bar a_{v'}]  -2\ell \frac{\epsilon}{\eta}\\
		&> \sum_{v' }  \zvec[a_0,\bar a_{v'}],
	\end{align*}
	where the last inequality follows from equation~\eqref{eq:contr}.
	 Hence, we reach a contradiction and $|V^2|\le \ell-\ell^{1-\delta}$.

	To conclude the proof, we show that $V^1$ is an independent set. Since $|V^1|\ge |V|-|V^2|=\ell^{1-\delta}$ we reach a contradiction.
	Let $v$ and $v'$ be two nodes in $V^1$ and w.l.o.g. let $\zvec[a_0,a_v]\ge \zvec[a_0,a_{v'}]$. We show that $(v,v')\notin E$. Since $v'\in V^1$, $\phi_{a_v}$ is not a safe deviation for player $1$ with respect to constraint $c^{v'}$. if $(v,v')\in E$, then
	\begin{align*}
	 \sum_{a^1 \in \A_1,a^2\in \A_2} \sum_{a \in\ \A_1} &\phi[a^1,a] \zvec[a^1,a^2] c_v(a,a^2) \\
	 &= \zvec[a_0,a_v']- \sum_{v'': (v'',v')\in E} \zvec[a_0,a_{v''}] -\frac{1}{4\ell} \zvec[a_0,a_F] c_v(a,a^2)\\
	 &\hspace{2cm}+  \sum_{a^1 \in \A_1 \setminus \{a_0\},a^2\in A_2} \sum_{a \in \A_1} \phi[a^1,a] \zvec[a^1,a^2] c_v(a,a^2)  \\
	  & \le \zvec[a_0,a_v'] -  \zvec[a_0,a_{v}] -\frac{1}{4\ell} \zvec[a_0,a_F] c_v(a,a^2)+  \\
	  & \hspace{2cm}+\sum_{a^1 \in \A_1 \setminus \{a_0\},a^2\in \A_2} \sum_{a \in \A_1} \phi[a^1,a] \zvec[a^1,a^2] c_v(a,a^2)\\
	  & \le -\frac{1}{4\ell} \zvec[a_0,a_F] c_v(a,a^2)+  \sum_{a^1 \in \A_1 \setminus \{a_0\},a^2\in \A_2} \sum_{a \in \A_1} \phi[a^1,a] \zvec[a^1,a^2] c_v(a,a^2)\\
	  & = c_v(z)\le 0.
	\end{align*}
	 Hence, $(v,v')\notin E$.
	 Since $V^1$ is an independent set of size at least $\ell^{1-\delta}$ we reach a contradiction. This concludes the proof.
\end{proof}
	\section{Proofs Omitted from Section~\ref{sec:optimal} }\label{sec:app_optimal}

\minmax*

\begin{proof}
	First, it is easy to see that
	\begin{align*}
		\sup_{\phi_i \in \Phi_i^\mcS(\zvec)} u_i(\phi_i \diamond \zvec)=
		\sup_{\phi_i\in\Phi_i} \inf_{\etavec_i \in \mathbb{R}^{m}_+} &\left(u_i(\phi_i \diamond \zvec)-\etavec_i^{\top} \cvec_i(\phi_i \diamond \zvec)\right).
	\end{align*}
	Indeed, for every $\phi_i \notin  \Phi_i^\mcS(\zvec)$, it holds that the vector $\cvec_i(\phi_i\diamond \zvec)$ has at least one positive component, and, thus, the vector of Lagrange multipliers $\etavec_i$ can be selected so that $u_i(\phi_i \diamond \zvec)-\etavec_i^{\top} \cvec_i(\phi_i \diamond \zvec)$ goes to $-\infty$.
	This implies that the supremum over $\Phi_i$ cannot be attained in $\Phi_i^\mcS(\zvec)$.
	%å
	On the other hand, for every $\phi_i \in  \Phi_i^\mcS(\zvec)$, all the components of $\cvec_i(\phi_i\diamond \zvec)$ are negative, and, thus, the $\inf$ is achieved by $\etavec_i=\zerovec$.
	This proves the first equality.
	
	Then, the second equality directly follows from the generalization of the max-min theorem for unbounded domains (see~\citep[Proposition~2.3]{ekeland1999convex}), which allows us to swap the $\sup$ and the $\inf$.
\end{proof}

\begin{lemma}\label{lem:lemmaa1}
	For any two real-valued functions $f(x)$ and $g(x)$ with $g(x)\le c$ then $\min(f(x), g(x))\le\min(f(x), c)$.
\end{lemma}

\begin{proof}
	We can identify three sets $I_1, I_2$ and $I_3$ defined as follows:
	\begin{align*}
		I_1&\coloneqq\{x\, s.t.\, f(x)\ge c\}\\
		I_2&\coloneqq\{x\, s.t.\, g(x)\le f(x)\le c\}\\
		I_3&\coloneqq\{x\, s.t.\, f(x)\le g(x)\le c\}.
	\end{align*}
	Then for all $x\in I_1$ we have that
	\(
	\min(f(x), c)=c\ge\min(f(x), g(x))=g(x),
	\)
	while for all $x\in I_2$ we have that
	\(
	\min(f(x), c)=f(x)\ge\min(f(x), g(x))=g(x).
	\)
	Finally for all $x\in I_3$ we have
	\(
	\min(f(x), c)=f(x)=\min(f(x), g(x))=f(x).
	\)
	In all three sets we have that $\min(f(x), c)\ge\min(f(x), g(x))$.
\end{proof}

\begin{lemma}\label{lem:lemmaa2}
	For all $\etavec_i\in\D^c$ it holds that
	\[
	\sup\limits_{\phi_i\in\Phi_i}\left(u_i(\phi_i\diamond\zvec)-\etavec_i^\top\cvec_i(\phi_i\diamond\zvec)\right)\ge 1
	\]
\end{lemma}

\begin{proof}
	Thanks to Assumption~\ref{ass:strictly} we have that for all $\zvec\in\Delta(\A)$ we have that there exists $\tilde\phi_i\in\Phi_i^\mcS(\zvec)$ such that $\cvec_i(\tilde\phi_i\diamond\zvec)\preceq-\rho\onevec$. Then, for all $\etavec_i\in\D^c$ we have:
	\[
	\etavec_i^\top \cvec_i(\tilde\phi_i\diamond \zvec)\le -\rho\|\etavec_i\|_1\le -1.
	\]
	This easily concludes the proof of the statement
	\[
	\sup\limits_{\phi_i\in\Phi_i}\left(u_i(\phi_i\diamond\zvec)-\etavec_i^\top\cvec_i(\phi_i\diamond\zvec)\right)\ge u_i(\tilde\phi_i\diamond\zvec)-\etavec_i^\top\cvec_i(\tilde\phi_i\diamond\zvec)\ge 1,
	\]
	as $u_i$ is positive.
\end{proof}

\begin{lemma}\label{lem:lemmaa3}
	For all $\etavec_i\in\D$ we have
	\[
	\inf\limits_{\etavec_i\in\D}\sup\limits_{\phi_i\in\Phi_i}\left(u_i(\phi_i\diamond\zvec)-\etavec_i^\top\cvec_i(\phi_i\diamond\zvec)\right)\le 1
	\]
\end{lemma}

\begin{proof}
	Since $u_i\le 1$ we have that
	\[
	\inf\limits_{\etavec_i\in\D}\sup\limits_{\phi_i\in\Phi_i}\left(u_i(\phi_i\diamond\zvec)-\etavec_i^\top\cvec_i(\phi_i\diamond\zvec)\right)\le 1-\sup\limits_{\etavec_i\in\D}\inf\limits_{\phi_i\in\Phi_i} \etavec_i^\top\cvec_i(\phi_i\diamond \zvec).
	\]
	Next we claim that $\sup\limits_{\etavec_i\in\D}\inf\limits_{\phi_i\in\Phi_i} \etavec_i^\top\cvec_i(\phi_i\diamond \zvec)\ge 0$. This follows from the fact that for all negative components of $\cvec_i(\phi_i\diamond \zvec)$ then the corresponding components of $\etavec_i$ will be $0$. This concludes the statement.
\end{proof}

\duality*

\begin{proof}
	In Lemma~\ref{lem:minmax} we already showed that:
	\begin{align*}
		\sup_{\phi_i \in \Phi_i^\mcS(\zvec)} u_i(\phi_i \diamond \zvec)&=
		\sup\limits_{\phi_i\in\Phi_i} \inf_{\etavec_i \in \mathbb{R}^{m}_+} \left(u_i(\phi_i\diamond \zvec)-\etavec_i^{\top} \cvec_i(\phi_i \diamond z)\right) \\
		&=\inf_{\etavec_i \in \mathbb{R}^{m}_+}  \sup\limits_{\phi_i\in\Phi_i}\left(u_i(\phi_i \diamond\zvec)-\etavec_i^{\top} \cvec_i(\phi_i \diamond\zvec)\right).
	\end{align*}
	Note that to prove the statement it is enough to prove that:
	\[
	\inf_{\etavec_i \in \mathbb{R}^{m}_+}  \sup\limits_{\phi_i\in\Phi_i}\left(u_i(\phi_i \diamond\zvec)-\etavec_i^{\top} \cvec_i(\phi_i \diamond\zvec)\right)=\inf_{\etavec_i \in\D}  \sup\limits_{\phi_i\in\Phi_i}\left(u_i(\phi_i \diamond\zvec)-\etavec_i^{\top} \cvec_i(\phi_i \diamond\zvec)\right)
	\]
	and more specifically that:
	\[
	\inf_{\etavec_i \in \mathbb{R}^{m}_+}  \sup\limits_{\phi_i\in\Phi_i}\left(u_i(\phi_i \diamond\zvec)-\etavec_i^{\top} \cvec_i(\phi_i \diamond\zvec)\right)\ge\inf_{\etavec_i \in\D}  \sup\limits_{\phi_i\in\Phi_i}\left(u_i(\phi_i \diamond\zvec)-\etavec_i^{\top} \cvec_i(\phi_i \diamond\zvec)\right)
	\]
	since the reverse inequality holds trivially.
	We can show this by the following inequalities:
	\begin{align*}
		\inf_{\etavec_i \in \mathbb{R}^{m}_+}  &\sup\limits_{\phi_i\in\Phi_i}\left(u_i(\phi_i \diamond\zvec)-\etavec_i^{\top} \cvec_i(\phi_i \diamond\zvec)\right)\\
		&=\min\left(\inf_{\etavec_i \in\D}  \sup\limits_{\phi_i\in\Phi_i}\left(u_i(\phi_i \diamond\zvec)-\etavec_i^{\top} \cvec_i(\phi_i \diamond\zvec)\right), \inf_{\etavec_i \in \D^c}  \sup\limits_{\phi_i\in\Phi_i}\left(u_i(\phi_i \diamond\zvec)-\etavec_i^{\top} \cvec_i(\phi_i \diamond\zvec)\right)\right)\\
		&\ge\min\left(\inf_{\etavec_i \in\D}  \sup\limits_{\phi_i\in\Phi_i}\left(u_i(\phi_i \diamond\zvec)-\etavec_i^{\top} \cvec_i(\phi_i \diamond\zvec)\right), 1\right)\\
		&=\inf_{\etavec_i \in\D}  \sup\limits_{\phi_i\in\Phi_i}\left(u_i(\phi_i \diamond\zvec)-\etavec_i^{\top} \cvec_i(\phi_i \diamond\zvec)\right),
	\end{align*}
	where the first inequality hold thanks to Lemma~\ref{lem:lemmaa1} and Lemma~\ref{lem:lemmaa2}, while that last equation follows from Lemma~\ref{lem:lemmaa3}. 
\end{proof}

\deltaOPTProblem*

\begin{proof}
	By definition we have that:
	$L_{\bar\D,\epsilon} =\ell(\tilde\zvec^\star)$, where $\tilde\zvec^\star$ is a solution to the problem
	\[
	\textnormal{P1}\coloneqq
	\begin{cases}
		\tilde\zvec^\star\in\arg\max\limits_{\zvec\in\mcS} &\ell(\zvec)\, s.t.\\
		&\epsilon+u_i(\tilde\zvec^\star)\ge\max\limits_{\phi_i\in\Phi_i}\left(u_i(\phi_i\diamond\tilde\zvec^\star)-\tilde\etavec_i^{\star,\top}\cvec_i(\phi_i\diamond\tilde\zvec^\star)\right)\\
		&\tilde\etavec_i^\star\in\arg\inf\limits_{\etavec_i\in\bar\D}\sup\limits_{\phi_i\in\Phi_i}\left(u_i(\phi_i\diamond\tilde\zvec^\star)-\etavec_i^{\top}\cvec_i(\phi_i\diamond\tilde\zvec^\star)\right)
	\end{cases}
	\]
	
	On the other hand, call $\zvec^\star$ the optimal Constrained Phi-equilibrium. This is a solution to the problem:
	\[
	\textnormal{P2}\coloneqq
	\begin{cases}
		\zvec^\star\in\arg\max\limits_{\zvec\in\mcS} &\ell(\zvec)\, s.t.\\
		&u_i(\zvec^\star)\ge\max\limits_{\phi_i\in\Phi_i}\left(u_i(\phi_i\diamond\zvec^\star)-\etavec_i^{\star,\top}\cvec_i(\phi_i\diamond\zvec^\star)\right)\\
		&\etavec_i^\star\in\arg\inf\limits_{\etavec_i\in\D}\sup\limits_{\phi_i\in\Phi_i}\left(u_i(\phi_i\diamond\zvec^\star)-\etavec_i^{\top}\cvec_i(\phi_i\diamond\zvec^\star)\right)
	\end{cases}
	\]
	which has value $L_{\D,0}=\ell(\zvec^\star)$.
	
	Moreover, thanks to Lemma~\ref{lem:duality} and since $\bar\D$ is $\delta$-optimal we have that:
	\[
	\max\limits_{\phi_i\in\Phi_i}\left(u_i(\phi_i\diamond\tilde\zvec^\star)-\tilde\etavec_i^{\star,\top}\cvec_i(\phi_i\diamond\tilde\zvec^\star)\right)\le\max\limits_{\phi_i\in\Phi_i}\left(u_i(\phi_i\diamond\zvec^\star)-\etavec_i^{\star,\top}\cvec_i(\phi_i\diamond\zvec^\star)\right)+\delta
	\]
	which implies that feasible correlated strategies of problem P2 are feasible  correlated strategies of problem P1, and thus problem P1 as long as $\delta\ge \epsilon$. Thus problem P1 is the problem of maximizing the same objective function over a larger set then problem P2 and thus $L_{\bar\D,\epsilon}\ge L_{\D,0}$.
\end{proof}

\gridisemopt*
\begin{proof}
	By Lemma \ref{lem:duality}, we know that for each player there exists an $\etavec_i^\star\in\D$ such that $\max_{\phi \in \Phi_i^\mcS(\zvec)} u_i(\phi_i\diamond \zvec)= \max\limits_{\phi_i\in\Phi_i}\left(u_i(\phi_i\diamond \zvec)-\etavec_i^{\star,\top} \cvec_i(\phi_i\diamond \zvec)\right)$.
	By construction of $\D_\epsilon$ there exists a $\bar \etavec_i \in \D_\epsilon $ such that $||\bar \etavec_i-\etavec_i^\star||_\infty \le \epsilon$.
	Thus
	\begin{align*}
		\max_{\phi \in \Phi_i^\mcS(\zvec)} u_i(\phi_i\diamond \zvec)&= \max\limits_{\phi_i\in\Phi_i}\left(u_i(\phi_i\diamond \zvec)-\etavec_i^{\star,\top} \cvec_i(\phi_i\diamond \zvec)\right)\\
		&\le \max\limits_{\phi_i\in\Phi_i}\left(u_i(\phi_i\diamond \zvec)-\bar \etavec_i^\top \cvec_i(\phi_i\diamond \zvec)\right)+ m \epsilon,
	\end{align*}
	where the last inequality comes the fact that:
	\[
	|(\etavec_i^\star-\bar\etavec_i)^\top\cvec_i(\phi_i\diamond\zvec)|\le \|\cvec_i(\phi_i\diamond\zvec)\|_1\|\etavec_i^\star-\bar\etavec_i\|_\infty\le m\epsilon
	\]
	as $\cvec_i\in [-1,1]^{m}$.
\end{proof}

\uniform*
\begin{proof}
	The proof exploits a probability interpretation of the Lagrange multipliers. Let $\etavec^\star$ be the optimal multipliers, \ie, $\etavec^\star\in\argmin_{\etavec\in\D}\max_{\phi_i\in\Phi_i}\left(u_i(\phi_i\diamond \zvec)-\etavec^{\top} \cvec_i(\phi_i\diamond \zvec)\right)$. Now consider a basis $\Gamma=\{\frac{1}{\rho}\evec_j\}_{j\in [m]}\cup\{\zerovec\}$ for $\D$.
	By Carathoedory's theorem there exists a distribution $\gamma\in\Delta(\Gamma)$ such that $\etavec^\star=\sum_{\etavec\in\Gamma}\gamma_{\etavec}\etavec$. Assume that $\epsilon$ and $\rho$ are such that $\sfrac{1}{\epsilon\rho}$ is an integer and take $\sfrac{1}{\rho\epsilon}$ samples from the distribution $\gamma$ and call $\tilde\etavec$ the resulting empirical mean.
	
	First, we argue that $\tilde\etavec\in\D_\epsilon$. Indeed $\tilde\etavec_j=\frac{k_j}{\sfrac{1}{\rho\epsilon}}\frac{1}{\rho}=\epsilon\left(\frac{k_j}{\sfrac{1}{\rho\epsilon}}\frac{1}{\rho\epsilon}\right)=\epsilon k_j$ where $k_j\in\mathbb{N}$ and thus we have that $\tilde\etavec\in\D_\epsilon$.\footnote{If $\epsilon$ if not such that $\sfrac{1}{\rho\epsilon}\in\mathbb{N}$ then the one can take $\lceil\sfrac{1}{\rho\epsilon}\rceil$ samples from $\gamma\in\Delta(\Gamma)$ and then the statement hold for a slightly smaller $\epsilon^\prime<\epsilon$ defined as $\epsilon^\prime:=\frac{1}{\lceil\sfrac{1}{\rho\epsilon}\rceil}\frac{1}{\rho}$.}
	
	Now we show that with high probability $\tilde\etavec\in\D_\epsilon$ is close (in terms of utilities) to the optimal multiplier $\etavec^\star$.
	First observe that:
	\begin{subequations}\label{eq:inequalitylemmaProbs}
	\begin{align}
	\etavec_i^{\star,\top}\cvec_i(\phi_i\diamond\zvec)&:=\sum\limits_{a_i\in\A^i,b_{i}\in\A^{i}}\left(\phi_i[b,a_i]\sum\limits_{\avec_{-i}\in\A^{-i}}\etavec^{\star,\top}\cvec_i(a_i,\avec_{-i})\zvec[b,\avec_{-i}]\right)\\
	&\le\sum\limits_{a_i\in\A^i,b_{i}\in\A^{i}}\left(\phi_i[b,a_i]\left(\delta_{a_i,b}+\sum\limits_{\avec_{-i}\in\A^{-i}}\tilde\etavec^{\top}\cvec_i(a_i,\avec_{-i})\zvec[b,\avec_{-i}]\right)\right)\\
	&=\sum\limits_{a_i\in\A^i,b_{i}\in\A^{i}}\left(\phi_i[b,a_i]\sum\limits_{\avec_{-i}\in\A^{-i}}\tilde\etavec^{\top}\cvec_i(a_i,\avec_{-i})\zvec[b,\avec_{-i}]\right)+\sum\limits_{a_i\in\A^i,b_{i}\in\A^{i}}\phi_i[b,a_i]\delta_{a_i,b}\\
	&=\tilde\etavec_i^{\top}\cvec_i(\phi_i\diamond\zvec)+\sum\limits_{a_i\in\A^i,b_{i}\in\A^{i}}\phi_i[b,a_i]\delta_{a_i,b}
	\end{align}
	\end{subequations}
	where the inequality comes from applying the Hoeffeding's inequality to every $a_i,b\in\A_i$:
	\[
	\left|\sum\limits_{a_{-i}\in\A^{-i}}\left(\tilde\etavec-\etavec^{\star}\right)^\top\cvec_i(a_i,a_{-i})\zvec[b,\avec_{-i}]\right|\le \delta_{a_i,b}
	\]
	where $\delta_{a_i,b}=\frac{2}{\rho}\sqrt{\frac{2}{\sfrac{1}{\rho\epsilon}}\log\left(\frac{2}{p_{a_i,b}}\right)}\left(\sum\limits_{\avec_{-i}\in\A^{-i}}\zvec[b,\avec_{-i}]\right)$ since the range of the each sample is $\frac{1}{\rho}\left(\sum_{\avec_{-i}\in\A^{-i}}\zvec[b,\avec_{-i}]\right)$.
	
	Moreover, for Hoeffeding's inequality, for every $a_i,b\in\A^i$ the above inequality holds with probability at least $1-p_{a_i,b}$ and thus holds for all the $a_i,b\in\A^i$ simultaneously, with probability at least $p:=\sum_{a_i,b\in\A^i}p_{a_i,b}$.
	If then we take $p_{a_i, b}:=\frac{1}{2|\A^i|^2}$ for all $a_i,b\in\A^i$, then we have that $p=1/2>0$ and $\delta:=\delta_{a_i,b}=\frac{2}{\rho}\sqrt{\frac{2}{\sfrac{1}{\rho\epsilon}}\log\left(|\A^i|\right)}\left(\sum\limits_{\avec_{-i}\in\A^{-i}}\zvec[b,\avec_{-i}]\right)$

	Now the following holds with probability at lest $1/2$:
	\[
	\left|\sum\limits_{a_{-i}\in\A^{-i}}\left(\tilde\etavec-\etavec^{\star}\right)^\top\cvec_i(a_i,a_{-i})\zvec[b,\avec_{-i}]\right|\le \delta\left(\sum_{\avec_{-i}\in\A^{-i}}\zvec[b,\avec_{-i}]\right),\quad\forall a_i,b\in\A^i
	\]
	
	The proof is concluded by observing plugging this definition of $\delta=\delta_{a_i,b}$ in Equation~\eqref{eq:inequalitylemmaProbs} yields
	\(
	\sum_{a_i\in\A^i,b_{i}\in\A^{i}}\phi_i[b,a_i]\delta_{a_i,b}=\delta,
	\)
	and we can conclude that:
	\[
	\etavec_i^{\star,\top}\cvec_i(\phi_i\diamond\zvec)\le\tilde\etavec_i^{\top}\cvec_i(\phi_i\diamond\zvec)+\delta.
	\]
	This holds with positive probability, and thus shows the existence of such $\tilde\etavec\in\D_\epsilon$ for which the above inequality holds and thus $\D_\epsilon$ is $\left(2\sqrt{\frac{2\epsilon}{\rho}\log(|\A^i|)}\right)$-optimal.
\end{proof}

	\section{Proofs Omitted from Section~\ref{sec:easy} }\label{sec:app_easy}

\convex*

\begin{proof}
	Let $\zvec^\prime$ and $\zvec^{\prime\prime}$ be Constrained $\epsilon$-Phi-equilibria that is for all $i\in[N]$:
	\[
	\epsilon+u_i(\zvec^\prime)\ge u_i(\phi_i^\prime\diamond\zvec^\prime)
	\]
	for $\phi^\prime\in\arg\max\limits_{\phi_i\in\Phi^\mcS_i}u_i(\phi_i\diamond\zvec^\prime)$. Equivalently it holds for all $i\in[N]$ that:
	\[
	\epsilon+u_i(\zvec^{\prime\prime})\ge u_i(\phi_i^{\prime\prime}\diamond\zvec^{\prime\prime})
	\]
	where $\phi^{\prime\prime}\in\arg\max\limits_{\phi_i\in\Phi^\mcS_i}u_i(\phi_i\diamond\zvec^{\prime\prime})$.
	For any $\zvec:=\alpha \zvec^\prime+(1-\alpha)\zvec^{\prime\prime}$ we have that:
	\begin{align*}
	\epsilon+u_i(\zvec)&=\alpha\left(\epsilon+ u_i(\zvec^\prime)\right)+(1-\alpha)\left(\epsilon+u_i(\zvec^{\prime\prime})\right)\\
	&\ge\alpha u_i(\phi_i^\prime\diamond\zvec^\prime)+(1-\alpha)u_i(\phi_i^{\prime\prime}\diamond\zvec^{\prime\prime})\\
	&\ge\max\limits_{\phi_i\in\Phi_i^\mcS} u_i(\phi_i\diamond\zvec),
	\end{align*}
	where the inequality holds for the linearity of $u_i$, the first inequality as both $\zvec^\prime$ and $\zvec^{\prime\prime}$ are Constrained $\epsilon$-Phi-equilibria and the last inequality holds since the $\max$ is a convex operator.
\end{proof}

\optimalsimple*

\begin{proof}
	\textsc{ApxCPE}$(1,0)$ amounts to solving the following problem:
	\begin{subequations}
		\begin{align}
		\max\limits_{\zvec\in\mcS} & \,\, \ell(\zvec) \quad \text{s.t.}  \\
		& u_i(\zvec) \ge \max_{\phi_i\in\Phi_i^\mcS} u_i(\phi_i\diamond \zvec) \quad \forall i\in \N,
		\end{align}
	\end{subequations}
	which can be written as an LP with (possibly) exponentially-many constraints, by writing a constraint for each vertex of $\Phi_i^\mcS$.
	We can find an exact solution to such an LP in polynomial time by means of the ellipsoid algorithm that uses suitable separation oracle.
	Such an oracle solves the following optimization problem for every $i\in \N$:
	\[
		\phi_i^\star\in\arg\max_{\phi_i\in\Phi_i^\mcS} u_i(\phi_i\diamond \zvec).
	\]
	Then, the oracle returns as a separating hyperplane the incentive constraint corresponding to a $\phi_i^\star$ (if any) such that $u_i(\zvec)\ge u_i(\phi_i^\star\diamond \zvec)$.
	Since all the steps of the separation oracle can be implemented in polynomial time, the ellipsoid algorithm runs in polynomial time, concluding the proof.
\end{proof}

\learning*

\begin{proof}
	Any regret minimizer $\mfR_i$ for $\Phi_i^\mcS$ guarantees that, for every $\phi_i \in \Phi_i^\mcS$:
	\begin{equation}\label{eq:regret1}
	\sum\limits_{t=1}^T u_i(\phi_i\diamond\zvec_t)- \sum\limits_{t=1}^T u_i(\phi_{i,t}\diamond\zvec_t) \le \epsilon_{i,T} \, T,
	\end{equation}
	where $\epsilon_{i,T}=o(T)$.
	Since $\xvec_{i,t}[a]=\sum_{b \in \A_i}\phi_{i,t}[b,a]\xvec_{i,t}[b]$ for all $ a\in\A_i$, for every $t \in [T]$ and $\avec = (a_i, \avec_{-i}) \in \A$:
	\begin{align*}\label{eq:fixedpoint}
	(\phi_{i,t}\diamond\zvec_t)[a_i,\avec_{-i}] &= \sum\limits_{b\in\A_i}\phi_{i,t}[b,a_i]\zvec[b,\avec_{-i}]\\
	&=\sum\limits_{b\in\A_i}\phi_{i,t}[b,a_i]\Big(\xvec_{i,t}[b]\otimes\xvec_{-i,t}[\avec_{-i}]\Big)\\
	&=\left(\sum\limits_{b\in\A_i}\phi_{i,t}[b,a_i]\xvec_{i,t}[b]\right)\otimes\xvec_{-i,t}[\avec_{-i}]\\
	&=\xvec_{i,t}[a_i]\otimes\xvec_{-i,t}[\avec_{-i}]\\
	&=\zvec_t[a_i,\avec_{-i}],
	\end{align*}
	Plugging the equation above into Equation~\eqref{eq:regret1}, we get:
	\[
	\sum\limits_{t=1}^T u_i(\phi_i\diamond\zvec_t)- \sum\limits_{t=1}^T u_i(\zvec_t) \le \epsilon_{i,T} \, T.
	\]
	Now, since $\bar\zvec_T:=\sum_{t=1}^T\zvec_t$ and $u_i(\zvec)$ is linear in $\zvec$, we can conclude that, for every $i\in \N$ and $\phi_i \in \Phi_i^\mcS$:
	\[
	u_i(\zvec_T)\ge u_i(\phi_i\diamond\bar\zvec_T) - \epsilon_{i,T},
	\]
	and, thus, by letting $\epsilon_{T} := \max_{i \in \N}\epsilon_{i,T}$ we get that $\bar\zvec_T$ satisfies the incentivize constrained for being a constrained $\epsilon_{T}$-Phi-equilibrium.
	%
	%is a constrained $(\epsilon_{T})$-Phi-equilibrium.
	%
	We are left to verify that $\bar\zvec_T\in\mcS$, namely $\cvec_i(\bar\zvec_T)\le \zerovec$ for all $i\in \N$.
	This readily proved as:
	\begin{align*}
	\cvec_i(\bar\zvec_T)&=\frac{1}{T}\sum\limits_{t=1}^T\cvec_i(\zvec_t)\\
	&=\frac{1}{T}\sum\limits_{t=1}^T\cvec_i(\phi_{i,t}\diamond\zvec_t)\\
	&=\frac{1}{T}\sum\limits_{t=1}^T\tilde\cvec_i(\phi_{i,t})\\
	&\leq \zerovec,
	\end{align*}
	where the first equality holds since $\cvec_i$ is linear, the second equality holds thanks to $\zvec_t = \phi_{i,t}\diamond\zvec_t$, the third one by Assumption~\ref{ass:cce}, while the inequality holds since $\phi_{i,t}\in\Phi_i^\mcS$.
	This concludes the proof of the first part of the statement.

	In conclusion, Algorithm~\ref{alg:noregret} runs in polynomial time as finding $\xvec_{i,t}[a]=\sum_{b\in\A_i}\phi_{i,t}[b,a_i]\xvec_{i,t}[b]$ for all $a \in\A_i$ is equivalent to finding a stationary distribution of a Markov Chain, which can be done in polynomial time.
	Moreover, we can implement the regret minimizers $\mfR_i$ over the polytopes $\Phi_i^\mcS$ so that their operations run in polynomial time, such as, \emph{e.g.}, \emph{online gradient descent}; see~\citep{hazan2016introduction}.
	%
	% it is well known that there exists no-regret algorithms $\mfR_i$ over polytopes that run in polynomial time in the dimension of the set $\Phi_i$.
	%
	%	We are left to verify that $\bar\zvec_T\in\mcS$, \ie $\Cvec_i(\bar\zvec_T)\le0$ for all $i\in[N]$. Consider the following chain of equations:
	%	%
	%	\begin{align*}
	%		\Cvec_i(\bar\zvec_T)&:=\frac{1}{T}\sum\limits_{t=1}^T\Cvec_i(\zvec_t)\\
	%		&=\frac{1}{T}\sum\limits_{t=1}^T\Cvec_i(\phi_{i,t}\diamond\zvec_t)\\
	%		&=\frac{1}{T}\sum\limits_{t=1}^T\tilde\Cvec_i(\phi_{i,t})\le0,
	%	\end{align*}
	%	%
	%	where the first equality holds since $\Cvec_i$ is linear, the second equality holds thanks to Equation~\ref{eq:fixedpoint}. Then we used that Assumption~\ref{ass:cce} and that $\phi_{i,t}\in\Phi_i^\mcS$.
	%	%
\end{proof}

\marginal*

\begin{proof}
	Since the costs $\cvec_i(\avec)$ of player $i \in \N$ only depends on player $i$'s action $a_i$ and \emph{not} on the actions of other players, it is possible to show that there exists $\tilde \cvec_i: \Phi_{\textnormal{CCE}}\to [-1,1]^m$ such that the following holds for every $\zvec \in \Delta_{\A}$:
	\begin{align*}
	\tilde\cvec_i(\phi_i) := \cvec_i(\phi_i\diamond\zvec).
	\end{align*}
	Indeed, for every $\phi_i \in \Phi_{\textnormal{CCE}}$, by definition of $\Phi_{\textnormal{CCE}}$ there exists a probability distribution $\boldsymbol{h} \in \Delta_{\A_i}: \phi_i[b,a]=\boldsymbol{h}[a]$ for all $b,a\in\A_i$.
	Then, for every $a_i \in \A_i$ and $\avec_{-i} \in \A_{-i}$, we can write:
	\begin{align*}
	(\phi_i\diamond \zvec)[a_i,\avec_{-i}] &= \sum_{b\in\A_i}\phi_i[b,a_{i}]\zvec[b,\avec_{-i}] \\
	& = \sum_{b\in\A_i}\boldsymbol{h}[a_i]\zvec[b,\avec_{-i}] \\
	& = \boldsymbol{h}[a_i] \sum_{b\in\A_i}\zvec[b,\avec_{-i}].
	\end{align*}
	Moreover, it holds:
	\begin{align*}
	\cvec_i (\phi_i\diamond \zvec)[a_i,\avec_{-i}] & = \sum_{\avec \in \A} \cvec_i(\avec) (\phi_i\diamond \zvec)[a_i,\avec_{-i}] \\
	& = \sum_{\avec \in \A} \cvec_i(\avec) \boldsymbol{h}[a_i] \sum_{b\in\A_i}\zvec[b,\avec_{-i}] \\
	& = \sum_{a_i \in \A_i} \cvec_i(a_i,\cdot) \boldsymbol{h}[a_i] \sum_{\avec_{-i} \in \A_{-i}}\sum_{b\in\A_i}\zvec[b,\avec_{-i}]\\
	& = \sum_{a_i \in \A_i} \cvec_i(a_i,\cdot) \boldsymbol{h}[a_i] ,
	\end{align*}
	which only depends on $\phi_i$, as desired.
	Notice that, in the equations above, for every $a \in \A_i$ we let $\cvec_i(a,\cdot)$ be the (unique) value of $\cvec_i(\avec)$ for all $\avec \in \A: a_i = a$.
	%
	% Since $\Cvec_i(\zvec)$ only depends on the marginalization of the $i$-th player we can write that:
	%
	%	\begin{align}
	%		\cvec_i(\phi_i\diamond\zvec) = \tilde\cvec^\prime_i(\tilde\xvec^{\phi_i}_i),
	%	\end{align}
	%	where $\tilde\xvec_i$ is the marginalized strategy obtained after deviation $\phi_i$, \ie $\tilde\xvec^{\phi_i}_i[a_i]:=\sum_{b\in\A^i}\phi_i[b,a_i]\xvec_i(b)$. Now, if $\phi_i\in\Phi_{i,,\textnormal{CCE}}$ we know it exists a probability vector $\cvec\in\Delta(\A_i)$ such that $\phi_i[b,a_i]=\cvec[a_i]$ for all $b,a_i\in\A^i$. Thus $\tilde\xvec^{\phi_i}_i$ becomes:
	%	\[
	%	\tilde\xvec^{\phi_i}_i[a_i]:=\sum_{b\in\A^i}\phi_i[b,a_i]\xvec_i(b)=\sum_{b\in\A^i}\cvec[a_i]\xvec_i(b) = \cvec[a_i]
	%	\]
	%	and thus $\Cvec_i(\phi_i\diamond\zvec)$ only depends on $\phi_i$ and Assumption~\ref{ass:cce} is verified.
\end{proof}

	%%%%%%%%%%%%%%%%%%%%%%%%%%%%%%%%%%%%%%%%%%%%%%%%%%%%%%%%%%%%%%%%%%%%%%%%%%%%%%%
	%%%%%%%%%%%%%%%%%%%%%%%%%%%%%%%%%%%%%%%%%%%%%%%%%%%%%%%%%%%%%%%%%%%%%%%%%%%%%%%
	% APPENDIX
	%%%%%%%%%%%%%%%%%%%%%%%%%%%%%%%%%%%%%%%%%%%%%%%%%%%%%%%%%%%%%%%%%%%%%%%%%%%%%%%
	%%%%%%%%%%%%%%%%%%%%%%%%%%%%%%%%%%%%%%%%%%%%%%%%%%%%%%%%%%%%%%%%%%%%%%%%%%%%%%%
	
	%%%%%%%%%%%%%%%%%%%%%%%%%%%%%%%%%%%%%%%%%%%%%%%%%%%%%%%%%%%%%%%%%%%%%%%%%%%%%%%
	%%%%%%%%%%%%%%%%%%%%%%%%%%%%%%%%%%%%%%%%%%%%%%%%%%%%%%%%%%%%%%%%%%%%%%%%%%%%%%%

\end{document}